\documentclass[journal]{IEEEtran}
\usepackage[noadjust]{cite}
\usepackage[caption=false,font=footnotesize]{subfig}
\usepackage{stfloats}
\usepackage{url}
\usepackage[colorlinks=true,linkcolor=black,citecolor=blue,urlcolor=blue]{hyperref}
\usepackage{amssymb,amsmath,amsthm}
\usepackage{array}
\usepackage{amsfonts}
\usepackage{graphicx}
  \graphicspath{{./Pics/}}
  \DeclareGraphicsExtensions{.pdf,.png}
\usepackage{xspace}
\usepackage{pgf}
\usepackage{tikz}
  \usetikzlibrary{arrows,shapes,fit,positioning}
  \usetikzlibrary{automata}
  \usetikzlibrary{backgrounds}
  \usetikzlibrary{decorations.pathmorphing}
  \usetikzlibrary{calc}
\usepackage[T1]{fontenc}
\usepackage[latin1]{inputenc}
\usepackage[english]{babel}
\usepackage{booktabs}
  \newcommand{\ra}[1]{\renewcommand{\arraystretch}{#1}}
\usepackage{graphbox}
\usepackage{outlines}
\usepackage{algpseudocode}
\usepackage[linesnumbered,ruled,noend,noline]{algorithm2e}
  \newcommand{\hrulealg}[0]{\vspace{1mm} \hrule \vspace{1mm}}
\usepackage{fixmath}
\usepackage[textwidth=1.9cm,color=green!10,textsize=footnotesize]{todonotes}

\newtheorem{thm}{Theorem}
\newtheorem{lem}[thm]{Lemma}

\newtheorem{cor}[thm]{Corollary}

\newtheorem{defn}[thm]{Definition}

\DeclareMathOperator{\A}{\mathcal A}
\DeclareMathOperator{\G}{\mathcal G}
\DeclareMathOperator{\B}{\mathcal B}

\DeclareMathOperator{\C}{\mathcal C}

\newcommand{\eps}{\varepsilon}
\newcommand{\complclass}[1]{{\sc #1}\xspace}

\newcommand{\PSpace}{\complclass{PSpace}}

\allowdisplaybreaks

\begin{document}
\title{K-Step Opacity in Discrete Event Systems: Verification, Complexity, and Relations}
\author{Ji{\v r}{\' i}~Balun and Tom{\' a}{\v s}~Masopust%
  \thanks{J. Balun and T. Masopust are with the Department of Computer Science, Faculty of Science, Palacky University in Olomouc, Czechia. E-mails: {\tt jiri.balun01{@}upol.cz, tomas.masopust{@}upol.cz}.}%
  \thanks{Partially supported by the Ministry of Education, Youth and Sports under the INTER-EXCELLENCE project LTAUSA19098 and by IGA PrF 2021 022.}}

\maketitle

\begin{abstract}
  Opacity is a property expressing whether a system may reveal its secret to a passive observer (an intruder) who knows the structure of the system but has a limited observation of its behavior. Several notions of opacity have been studied, including current-state opacity, K-step opacity, and infinite-step opacity.
  We study K-step opacity that generalizes both current-state opacity and infinite-step opacity, and asks whether the intruder cannot decide, at any time, whether or when the system was in a secret state during the last K observable steps.
  We design a new algorithm deciding K-step opacity the complexity of which is lower than that of existing algorithms and that does not depend on K.
  We then compare K-step opacity with other opacity notions and provide new transformations among the notions that do not use states that are neither secret nor non-secret (neutral states) and that are polynomial with respect to both the size of the system and the binary encoding of K.
\end{abstract}
\begin{IEEEkeywords}
  K-Step Opacity, Discrete event systems, Verification, Complexity, Transformations
\end{IEEEkeywords}

\section{Introduction}
  \IEEEPARstart{P}{roperties} that guarantee to keep some information in a system secret include anonymity~\cite{SchneiderS96}, noninterference~\cite{Hadj-Alouane05}, secrecy~\cite{Alur2006}, security~\cite{Focardi94}, and opacity~\cite{Mazare04}.
  In this paper, we are interested in opacity of systems modeled by finite automata.

  Opacity is an information flow property asking if a system prevents an intruder from revealing the secret. The intruder is a passive observer that knows the structure of the system but has only limited observations of its behavior. Intuitively, the intruder estimates the behavior of the system, and the system is opaque if for every secret behavior, there is a non-secret behavior that looks the same to the intruder.

  There are two common ways to model the secret: a set of secret states and a set of secret behaviors.
  In the former case, the opacity is referred to as state-based, introduced by Bryans et al.~\cite{Bryans2005,BryansKMR08} for systems modeled by Petri nets and transition systems, and later adapted to (stochastic) automata by Saboori and Hadjicostis~\cite{SabooriHadjicostis2007}.
  In the latter case, the opacity is referred to as language-based, introduced by Badouel et al.~\cite{Badouel2007} and Dubreil et al.~\cite{Dubreil2008}.
  For more details, see Jacob et al.~\cite{JacobLF16}.

  Several opacity notions have been studied in the literature, including language-based opacity (LBO), initial-state opacity (ISO), initial-and-final-state opacity (IFO), current-state opacity (CSO), K-step opacity (K-SO), and infinite-step opacity (INSO).
  While initial-state opacity prevents the intruder from revealing, at any step of the computation, whether the system started in a secret state, current-state opacity prevents the intruder only from revealing whether the current state of the system is secret.
  The intruder may, however, realize in the future that the system was in a secret state at a former step of the computation. For example, if the intruder estimates that the system is in one of two states and, in the next step, the system proceeds by an observable event that is possible only from one of the states, then the intruder reveals the state in which the system was one step ago.

  This issue led to the introduction of K-step opacity~\cite{SabooriHadjicostis2007,SabooriH12a}. K-step opacity requires that the intruder cannot ascertain the secret in the current state and K subsequent observable steps. Two special cases for K $=0$ and K $=\infty$ are know as current-state opacity and infinite-step opacity, respectively, though the notion of infinite-step opacity may be confusing for finite automata, because an automaton with $n$ states is infinite-step opaque if and only if it is $(2^n-2)$-step opaque~\cite{Yin2017}.

  The complexity of known algorithms deciding K-step opacity depends on K. For example, the two-way observer of Yin and Lafortune~\cite{Yin2017} has complexity $O(\min\{n2^{2n},n\ell^{K} 2^{n}\})$, including a minor correction by Lan et al.~\cite{Lan2020}, where $n$ is the number of states of the automaton and $\ell$ is the number of observable events. Obviously, the complexity depends on K if $\ell^K < 2^n$.
  We recently designed an algorithm with complexity $O((K+1)2^n (n + m\ell^2))$, where $m \le \ell n^2$ is the number of transitions in the projected NFA, which is faster than the two-way observer if K is larger than $2^n-2$ or polynomial in $n$~\cite{BalunMasopust2021}. The reader can find more methods with their experimental comparisons in Wintenberg et al.~\cite{Wintenberg2021}.

  In this paper, we further improve the complexity of deciding K-step opacity to $O((n+m)2^n)$, which does not depend on K. We then provide new transformations among K-step opacity, current-state opacity, and infinite-step opacity. These transformations have been studied by Balun and Masopust~\cite{BalunMasopust2021}, who have shown that the notions are transformable to each other in polynomial time, and the results do not have more observable events and preserve determinism. However, the transformations from K-step opacity are polynomial only if K is small or considered as constant, whereas a large value of K makes the transformations infeasible. In addition, some of the transformations use {\em neutral states}---states that are neither secret nor non-secret.

  Here we suggest new transformations that are polynomial in both the size of the system and a binary encoding (logarithm) of K, and that do not use neutral states.
  Wu and Lafortune~\cite{WuLafortune2013} studied transformations among other notions of opacity. We refer the reader to Balun and Masopust~\cite{BalunMasopust2021} for an overview of these transformations and the complexity results.

\section{Preliminaries}
  We assume that the reader is familiar with discrete-event systems~\cite{Lbook}. For a set $S$, $|S|$ denotes the cardinality of $S$ and $2^{S}$ its power set. An alphabet $\Sigma$ is a finite nonempty set of events. A string over $\Sigma$ is a sequence of events; the empty string is denoted by $\varepsilon$. The set of all finite strings over $\Sigma$ is denoted by $\Sigma^*$. A language $L$ over $\Sigma$ is a subset of $\Sigma^*$. The set of prefixes of strings of $L$ is the set $\overline{L}=\{u \mid \text{there is } v\in \Sigma^* \text{ such that } uv \in L\}$. For $u \in \Sigma^*$, $|u|$ is the length of $u$.

  A {\em nondeterministic finite automaton\/} (NFA) over an alphabet $\Sigma$ is a structure $\G = (Q,\Sigma,\delta,I,F)$, where $Q$ is a finite  set of states, $I\subseteq Q$ is a set of initial states, $F \subseteq Q$ is a set of marked states, and $\delta \colon Q\times\Sigma \to 2^Q$ is a transition function that can be extended to the domain $2^Q\times\Sigma^*$ by induction.
  For a set $Q_0\subseteq Q$, the set $L_m(\G,Q_0) = \{w\in \Sigma^* \mid \delta(Q_0,w)\cap F \neq\emptyset\}$ is the language marked by $\G$ from the states of $Q_0$, and $L(\G,Q_0) = \{w\in \Sigma^* \mid \delta(Q_0,w)\neq\emptyset\}$ is the language generated by $\G$ from $Q_0$. The languages {\em marked\/} and {\em generated\/} by $\G$ are $L_m(\G)=L_m(\G,I)$ and $L(\G)=L(\G,I)$, respectively.
  For $S\subseteq \Sigma^*$, we write $\delta(Q,S) = \cup_{s\in S}\,\delta(Q, s)$.
  The NFA $\G$ is {\em deterministic\/} (DFA) if $|I|=1$ and $|\delta(q,a)|\le 1$ for every $q\in Q$ and $a \in \Sigma$.

  A {\em discrete-event system\/} (DES) $G$ over $\Sigma$ is an NFA over $\Sigma$ together with the partition of $\Sigma$ into $\Sigma_o$ and $\Sigma_{uo}$ of {\em observable\/} and {\em unobservable events}, respectively. If the marked states are irrelevant, we omit them and simply write $G=(Q,\Sigma,\delta,I)$.

  State estimation is modeled by {\em projection\/} $P\colon \Sigma^* \to \Sigma_o^*$, which is a morphism defined by $P(a) = \varepsilon$ if $a\in \Sigma_{uo}$, and $P(a)= a$ if $a\in \Sigma_o$. The action of $P$ on a string $a_1\cdots a_n$ is to erase unobservable events: $P(a_1\cdots a_n)=P(a_1) \cdots P(a_n)$. The definition can be readily extended to languages.

  Let $G$ be a DES over $\Sigma$ with projection $P\colon \Sigma^* \to \Sigma_o^*$. The {\em projected automaton\/} of $G$ is the NFA $P(G)$ obtained from $G$ by replacing every transition $(p,a,q)$ by $(p,P(a),q)$, and by eliminating the $\eps$-transitions. In particular, if $\delta$ is the transition function of $G$, then the transition function $\gamma\colon Q\times \Sigma_o \to 2^Q$ of $P(G)$ is defined as $\gamma(q,a)=\delta(q,P^{-1}(a))$. Then, $P(G)$ is an NFA over $\Sigma_o$ with the same states as $G$ that recognizes the language $P(L_m(G))$ and that can be constructed in polynomial time~\cite{HopcroftU79}.
  The DFA constructed from $P(G)$ by the standard subset construction is called an {\em observer of $G$}~\cite{Lbook}, which has up to exponentially more states than $G$~\cite{JiraskovaM12,wong98}.

\section{K-Step Opacity and its Verification}\label{ksoAlgo}
  We denote the set of non-negative integers by $\mathbb{N}$. For $K\in\mathbb{N}_\infty = \mathbb{N}\cup\{\infty\}$, K-step opacity asks if the intruder cannot reveal the secret in the current and $K$ subsequent states.

  \begin{defn}
    Given a DES $G=(Q,\Sigma,\delta,I)$ and K~$\in\mathbb{N}_\infty$. System $G$ is {\em K-step opaque (K-SO)} w.r.t. secret states $Q_S$, non-secret states $Q_{NS}$, and $P\colon \Sigma^*\to\Sigma_o^*$ if for every string $st \in L(G)$ with $|P(t)| \leq K$ and $\delta(\delta(I, s)\cap Q_S, t) \neq \emptyset$, there is $s't' \in L(G)$ such that $P(s)=P(s')$, $P(t)=P(t')$, and $\delta (\delta(I,s')\cap Q_{NS}, t') \neq \emptyset$.
  \end{defn}

  Two special cases of K-step opacity include 0-step opacity also known as {\em current-state opacity (CSO)}, and $\infty$-step opacity aka {\em infinite-step opacity (INSO)}~\cite{SabooriH12a}, which, for a DES with $n$ states, coincides with $(2^n-2)$-step opacity~\cite{Yin2017}.

  The complexity of existing algorithms verifying K-SO is exponential and depends on K. Exponential complexity seems unavoidable because the problem is \PSpace-complete~\cite{BalunMasopust2021}.
  We now design an algorithm verifying K-SO with complexity $O((n+m)2^n)$, where $n$ is the number of states of the automaton and $m$ is the number of transitions of the projected NFA, which improves the existing complexity and does not depend on K. Comparing the complexity with that of Wintenberg et al.~\cite{Wintenberg2021}, who neglect the number of transitions in the automata, our complexity can be stated as $O(n2^n)$, which is better than the results in Wintenberg et al.~\cite{Wintenberg2021}.

  \begin{algorithm}[h]
    \DontPrintSemicolon
    \caption{Verification of K-step opacity}
    \label{alg1}
      \SetKwInOut{Input}{Input}
      \SetKwInOut{Output}{Output}
      \Input{A DES $G=(Q,\Sigma,\delta,I)$, $Q_S,Q_{NS}\subseteq Q$, $\Sigma_o\subseteq \Sigma$, and K~$\in\mathbb{N}_\infty$.}
      \Output{{\tt true} if and only if $G$ is K-SO w.r.t. $Q_S$, $Q_{NS}$, and $P\colon \Sigma^*\to\Sigma_o^*$}
      \hrulealg
      Set $Y := \emptyset$\;
      Compute the observer $G^{obs}$ of $G$\;
      Compute the projected automaton $P(G)$ of $G$\;
      \For{every reachable state $X$ of $G^{obs}$}{
        \For{every state $x\in X\cap Q_S$}{
          add state $(x,X \cap Q_{NS})$ to set $Y$\label{line6}
        }
      }
      Compute the product automaton $\C = P(G) \times G^{obs}$ with the states of $Y$ as initial states\;
      Use BFS to mark states of $\C$ reachable from $Y$ in at most $K$ steps\;
      \lIf {$\C$ contains a marked state of the form $(q,\emptyset)$}{
        \Return {\tt false} {\bf else} \Return {\tt true} }
  \end{algorithm}

  Our algorithm is described as Algorithm~\ref{alg1}. Intuitively, we compute the observer of $G$ (on demand also some of its non-reachable states), the projected NFA of $G$, and their product automaton $\C$. For every reachable state $X$ of the observer, we make the states $(x,X\cap Q_{NS})$, where $x$ is a secret state from $X$ and the second component is a set of all non-secret states from $X$, initial in $\C$. Then we use Breadth-First Search (BFS)~\cite{IntroToAlg} to search $\C$ and to mark all states of $\C$ that are reachable in at most K steps from an initial state. This is done as follows.
  First, we push all initial states of $\C$ to the queue, followed by pushing number 0 (in binary) to the queue. After processing the initial states, we remove 0 from and push 1 to the queue. At this point, the queue contains all states of $\C$ reachable from the initial states in one step, followed by number 1. The algorithm proceeds this way until it has either visited all states of $\C$ or the number stored in the queue is K. All and only visited states of $\C$ are marked. We show in Theorem~\ref{thm-correctAlg1} that $G$ is K-SO if and only if no state of the form $(\cdot,\emptyset)$ is marked in $\C$.

  Before that, we illustrate Algorithm~\ref{alg1} by considering one-step opacity of the DES $G$ depicted in Figure~\ref{fig:inso-cso-ex} where all events are observable, $Q_S=\{2\}$, and $Q_{NS}=\{4\}$. A relevant part of the observer $G^{obs}$ is depicted in the same figure. Since $G$ has no unobservable events, $P(G)=G$. The only reachable state $X=\{2,4\}$ of $G^{obs}$ intersecting $Q_S$ results in $Y=\{(2,\{4\})\}$. The marked part of $\C_1 = P(G) \times G^{obs}$ reachable from $Y$ in at most one step is depicted in Figure~\ref{ex:alg1-obs}. Since state $(3,\emptyset)$ is marked in $\C_1$, $G$ is not one-step opaque; indeed, observing $ab$, the intruder reveals that $G$ was in a secret state.

  \begin{figure}
    \centering
    \includegraphics[align=c,scale=.8]{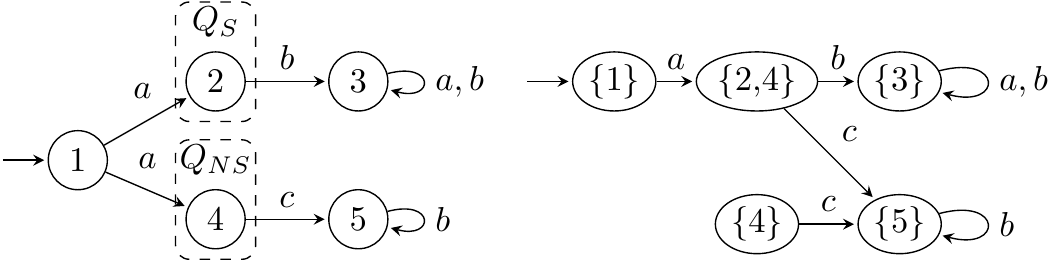}
    \caption{A DES $G$ (left) and a relevant part of the observer $G^{obs}$ (right).}
    \label{fig:inso-cso-ex}
  \end{figure}

  \begin{figure}
    \centering
    \includegraphics[scale=.7]{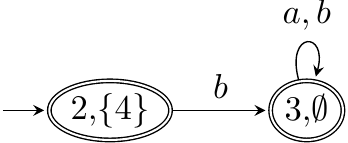}\qquad
    \includegraphics[scale=.7]{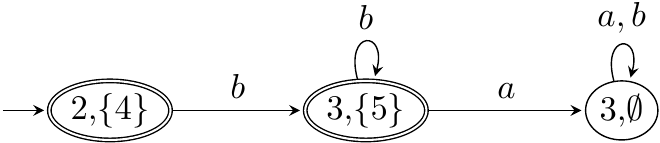}
    \caption{The reachable parts of $\C_1$ (left) and $\C_2$ (right).}
    \label{ex:alg1-obs}
    \label{ex:alg1-obs3}
  \end{figure}

  We now assume that event $c$ is unobservable, and denote $G$ with $a,b$ observable, $c$ unobservable, $Q_S=\{2\}$, $Q_{NS}=\{4\}$ by $\tilde{G}$. The automaton $P(\tilde{G})$ and a relevant part of $\tilde G^{obs}$ are depicted in Figure~\ref{ex:alg1-obs2}. The only reachable state $X=\{2,4,5\}$ of $\tilde{G}^{obs}$ intersecting $Q_S$ results in $Y=\{(2,\{4\})\}$. The marked part of $\C_2 = P(\tilde{G}) \times \tilde{G}^{obs}$ is depicted in Figure~\ref{ex:alg1-obs3}. Since no state of the form $(\cdot,\emptyset)$ is marked in $\C_2$, $\tilde{G}$ is one-step opaque.

  \begin{figure}
    \centering
    \includegraphics[scale=.8]{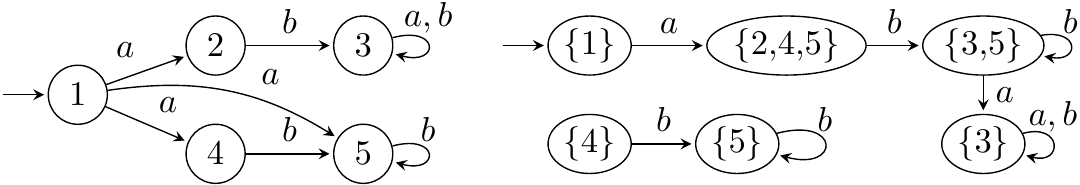}
    \caption{Projected automaton $P(\tilde{G})$ (left) and a relevant part of $\tilde G^{obs}$ (right).}
    \label{ex:alg1-obs2}
  \end{figure}

  We now prove the correctness of our algorithm.
  \begin{thm}\label{thm-correctAlg1}
    A DES $G$ is K-SO w.r.t. $Q_S$, $Q_{NS}$, and $P$ if and only if Algorithm~\ref{alg1} returns {\tt true}.
  \end{thm}
  \begin{proof}
    If $G=(Q,\Sigma,\delta,I)$ is not K-SO, then there exists $st\in L(G)$ such that $|P(t)|\le K$, $\delta(\delta(I,s)\cap Q_S, t) \neq \emptyset$, and $\delta(\delta(I,P^{-1}P(s))\cap Q_{NS},P^{-1}P(t)) = \emptyset$. We have two cases.
    (i) If $\delta(I,P^{-1}P(s))\cap Q_{NS} = \emptyset$, then $G$ is not K-SO. The algorithm detects this situation by $X=\delta(I,P^{-1}P(s))$, since there is $q\in X\cap Q_{S} \neq \emptyset$ and $X\cap Q_{NS} = \emptyset$, which results in adding $(q,\emptyset)$ to $Y$ in line~\ref{line6}.
    (ii) If $\delta(I,P^{-1}P(s))\cap Q_{NS} = Z \neq \emptyset$, then the pairs $(\delta(I,P^{-1}P(s))\cap Q_S)\times \{Z\}$ are added to $Y$. Since $\delta(\delta(I,s)\cap Q_S, t) \neq \emptyset$, there is $(z,Z)\in  Y$ such that $P(t)$ leads the automaton $P(G)$ from state $z$ to a state $q$. However, $\delta(Z,P^{-1}P(t))=\emptyset$ implies that $P(t)$ leads the observer of $G$ from state $Z$ to state $\emptyset$, and hence $(q,\emptyset)$ is reachable in $\C$ from a state of $Y$ in at most $|P(t)|\le K$ steps.

    On the other hand, if $G$ is K-SO, we show that no state of the form $(q,\emptyset)$ is reachable in $\C$ from $Y$ in at most K steps. For the sake of contradiction, assume that a state $(q,\emptyset)$ is marked in $\C$. Then, there must be a string $s$ such that $\delta(I,P^{-1}P(s))=X$ in $G$, that is, $P(s)$ reaches state $X$ in the observer of $G$, and there is $z \in X \cap Q_S$, $X \cap Q_{NS} = Z$, $(z,Z)\in Y$, and state $(q,\emptyset)$ is reached from state $(z,Z)$ in $\C$ by a string $w\in \Sigma_o^*$ of length at most K. In particular, there is $t\in P^{-1}(w)$ moving $G$ from state $z$ to state $q$. But then $q\in \delta(\delta(I,s)\cap Q_S,t)\neq\emptyset$, and $\delta(\delta(I,P^{-1}P(s))\cap Q_{NS},P^{-1}(w)) = \delta(Z,P^{-1}(w)) = \emptyset$, which means that $G$ is not K-SO---a contradiction.
  \end{proof}

  Finally, we discuss the complexity of our algorithm.
  \begin{thm}
    The space and time complexity of Algorithm~\ref{alg1} is resp. $O(n2^n)$ and $O((n + m)2^n)$, where $n$ is the number of states of $G$ and $m$ is the number of transitions of $P(G)$. Further, $m \le \ell n^2$, where $\ell$ is the number of observable events.
  \end{thm}
  \begin{proof}
    Computing the observer and the projected NFA, lines~2 and~3, takes time $O(\ell 2^n)$ and $O(m+n)$, resp. The cycle on lines~4--6 takes time $O(n 2^n)$. Constructing $\C$, line 7, takes time $O(n 2^n + m 2^n)$, where $O(n 2^n)$ is the number of states and $O(m 2^n)$ is the number of transitions of $\C$. The BFS takes time linear in $\C$, and the condition of line 9 can be processed during the BFS. Since $m\ge \ell$, the proof is complete.
  \end{proof}

\section{Relation to other Opacity Notions}
  We now design polynomial-time transformations of K-SO to CSO, and vice versa. For the transformations of CSO to other opacity notions, we refer the reader to the literature~\cite{WuLafortune2013,BalunMasopust2021}.
  Compared with the transformations of Wu and Lafortune~\cite{WuLafortune2013} and Balun and Masopust~\cite{BalunMasopust2021}, which use neutral states and are polynomial in the system size and the value of K, our new transformations do not use neutral states and are polynomial in the size of the system and the encoding (logarithm) of K.

  The need for the new transformations comes from the facts that (i) a large value of K makes the existing transformations infeasible, and (ii) the meaning of neutral states is unclear or questionable. Although we allow neutral states to appear in the systems, we neither use them nor create them in the transformations; using neutral states would result in transformations that do not work when the neutral states are not allowed~\cite{BalunMasopust2021}.

  Our motivation for the transformations is two-fold. First, they provide a deeper understanding of differences/similarities of the notions: we see that one secret state is sufficient for CSO, and we learn how to transform K-SO to K'-SO for any K and K'.
  Second, they are a tool to transfer complexity results among the notions: we get that deciding CSO for systems with a single secret state is as hard as deciding CSO for general systems, or that the existing complexity results for K-SO (and hence also INSO and CSO)~\cite[Table~1]{BalunMasopust2021} hold for systems without neutral states and K given as part of the input.

  To simplify the presentation of this section, some auxiliary technical results are moved to the appendices.

\subsection{Transforming CSO to K-SO for any K~$\in\mathbb{N}_\infty$}\label{redCSOtoKSO}
  The problem of deciding current-state opacity consists of a DES $G=(Q, \Sigma, \delta, I)$, secret states $Q_{S}$, non-secret states $Q_{NS}$, and projection $P\colon\Sigma^*\to \Sigma_o^*$.
  From $G$, we construct a DES $G'=(Q\cup\{q_{s},q_{ns}\}, \Sigma\cup\{@\}, \delta', I)$ over the alphabet $\Sigma\cup\{@\}$, where $@$ is a new observable event, by adding two new states $q_s$ and $q_{ns}$. The transition function $\delta'$ of $G'$ is initialized as the transitions function $\delta$ of $G$ and further extended as follows, see Figure~\ref{fig:cso-inso} for an illustration:
  \begin{enumerate}
    \item for every state $q\in Q_{S}$, we add $(q,@,q_{s})$ to $\delta'$;
    \item for every state $q\in Q_{NS}$, we add $(q,@,q_{ns})$ to $\delta'$.
  \end{enumerate}

  We define $P'\colon (\Sigma\cup\{@\})^* \to (\Sigma_o\cup\{@\})^*$, secret states $Q_S'=\{q_{s}\}$, and non-secret states $Q_{NS}'=Q_S\cup Q_{NS}\cup\{q_{ns}\}$.

  \begin{figure}
    \centering
    \includegraphics[align=c,scale=.8]{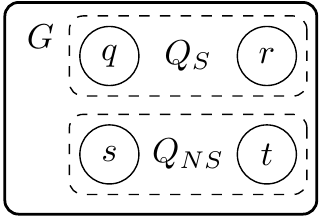}
    $\Longrightarrow$
    \includegraphics[align=c,scale=.8]{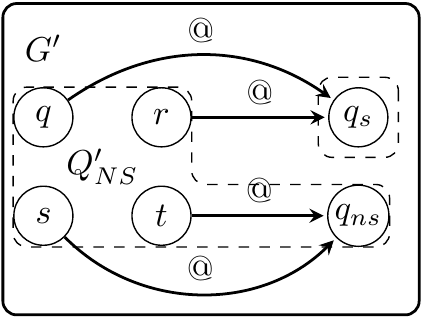}
    \caption{Transforming CSO to K-SO for any K~$\in\mathbb{N}_\infty$.}
    \label{fig:cso-inso}
  \end{figure}

  We now prove the correctness of the transformation.
  \begin{thm}\label{thm:cso-inso}
    The DES $G$ is CSO w.r.t. $Q_S$, $Q_{NS}$, $P$ iff the DES $G'$ is K-SO w.r.t. $Q_S'$, $Q_{NS}'$, $P'$.
  \end{thm}
  \begin{proof}
    Assume that $G$ is not CSO. Then there is a string $w\in\Sigma^*$ that leads $G$ to a secret state, and every string that looks the same as $w$ leads $G$ out of non-secret states. Then, in $G'$, generating the string $w@$ ends up in the secret state $q_s \in \delta'(I,w@)\cap Q_{S}' \neq \emptyset$. Since generating any string that looks the same as $w$ leads $G$ to a state out of non-secret states, we have that $\delta'(I,P'^{-1}P'(w@))\cap Q_{NS}' = \emptyset$. Therefore, $G'$ is not CSO, neither K-SO for any K~$\in\mathbb{N}_\infty$.

    Now, assume that $G$ is CSO, and let $st \in L(G')$ be such that $s$ leads $G'$ to a secret state and $t$ may be generated from this secret state in $G'$, formally $\delta'(\delta'(I, s)\cap Q_S', t) \neq \emptyset$. Then, $s=s_1@$ where $s_1$ does not contain $@$, and $t=\eps$.
    By construction, generating $s_1$ in $G$ ends up in a secret state. Since $G$ is CSO, there is a string $s_1'\in P^{-1}P(s_1)$ looking the same as $s_1$ such that generating $s_1'$ in $G$ ends up in a non-secret state. Then, generating $s_1'@$ in $G'$ ends up in a non-secret state, and hence taking $s'=s_1'@$ gives that $P'(s't) = P'(s') = P'(s) = P'(st)$ and $\delta' (\delta'(I,s')\cap Q_{NS}', t) \neq \emptyset$, which shows that $G'$ is K-SO for any K~$\in\mathbb{N}_\infty$.
  \end{proof}

  The transformation can be done in polynomial time, does not depend on K, and does not use neutral states. However, it introduces a new observable event.

  To decrease the number of observable events in $G'$, we may notice that $G'$ is K-SO, for any K~$\in \mathbb{N}_\infty$, if and only if $G'$ is CSO, since there are no transitions from the secret state $q_s$ of $G'$.
  Taking an encoding $e'\colon \Sigma_o \to \{0,1\}^k$ for a suitable $k$ (see Appendix~\ref{appA}), and defining $e(a) = 0 e'(a)$, for $a\in\Sigma_o$, and $e(@)=1^{k+1}$, we get an encoding $e\colon \Sigma_o\cup\{@\} \to \{0,1\}^{k+1}$ that encodes observable events of $\Sigma$ as binary sequences starting with $0$, and $@$ as a sequence of $1$'s.
  Applying the construction of Appendix~\ref{appA} to $G'$ and $e$ results in $G''$ with two observable events, $0$ and $1$, the only secret state $q_s$, and no transitions from the secret state $q_s$, $G''$ is K-SO if and only if $G''$ is CSO, which is if and only if $G'$ is CSO by Lemma~\ref{binEnc}.
  
  Notice that the transformation reduces CSO to K-SO with a single secret state, and hence we have the following corollary.
  \begin{cor}
    For any K~$\in\mathbb{N}_\infty$, deciding K-step opacity of a system with a single secret state and two or more observable events is \PSpace-complete.
    \qed
  \end{cor}

  If $G$ has a single observable event, the previous construction results in $G'$ with two observable events, and the construction of $G''$ does not work because the technique of Appendix~\ref{appA} requires at least three observable events in $G'$. Therefore, we design a direct transformation preserving a single observable event that does not admit neutral states. For systems admitting neutral states, we refer to our recent work~\cite{BalunMasopust2021}.

  The problem of deciding CSO for systems with a single observable event consists of a DES $G=(Q,\Sigma,\delta,I)$ with $\Sigma_o=\{a\}$, secret states $Q_{S}$, non-secret states $Q_{NS} = Q \setminus Q_{S}$, and projection $P\colon\Sigma^*\to \{a\}^*$.
  From $G$, we construct a DES $G'=(Q\cup\{q_0^\star,q_1^\star,q_2^\star\}, \Sigma\cup\{u\}, \delta', I)$ by adding a new unobservable event $u$ and three new states $q_0^\star,q_1^\star,q_2^\star$.
  The transition function $\delta'$ is initialized as $\delta$ and further extend by adding, for each $q\in Q_{NS}$, the transition $(q,u,q_1^\star)$, and by adding three transitions $(q_1^\star,a,q_2^\star)$, $(q_0^\star,a,q_0^\star)$, and $(q_2^\star,a,q_2^\star)$; see Figure~\ref{fig:cso-inso-single} for an illustration.
  Now, we determine (in linear time) whether the language $L(G)$ is finite. If so, we denote by
  $
    m=\max\{|P(w)| \mid w\in L(G)\}
  $
  the maximal number of observable events in the strings of $L(G)$, and by
  $
    Q_{max}=\{q\in Q\mid q\in\delta(I,P^{-1}(a^m))\}
  $
  the states reachable by the strings with the maximal number of observable events. Finally, we add the transition $(q,a,q_0^\star)$, for every $q\in Q_{max}$, to $\delta'$ and define $P'\colon (\Sigma\cup\{u\})^* \to \Sigma_o^*$, secret states $Q_S'=Q_S\cup\{q_2^\star\}$, and non-secret states $Q_{NS}'=Q_{NS}\cup\{q_0^\star,q_1^\star\}$, see Figure~\ref{fig:cso-inso-single}.

  \begin{figure}
    \centering
    \includegraphics[align=c,scale=.8]{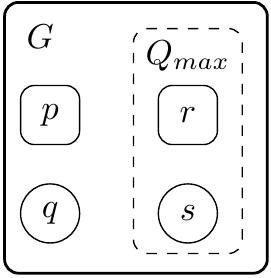}
    $\Longrightarrow$
    \includegraphics[align=c,scale=.8]{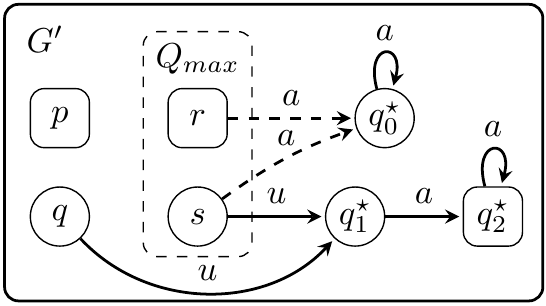}
    \caption{Transforming CSO to K-SO with a single observable event. Secret states are squared; the dashed transitions are in $G'$ only if $L(G)$ is finite.}
    \label{fig:cso-inso-single}
  \end{figure}

  We first formulate a simple, but important, observation.
  \begin{lem}\label{lemmaA}
    If $G$ is CSO, then, for every $w\in L(G)$, there exists $w'\in P^{-1}P(w)$ such that $\delta(I,w')\cap Q_{NS} \neq \emptyset$.
  \end{lem}
  \begin{proof}
    For $w\in L(G)$, either $\delta(I,w)\cap Q_{NS} \neq \emptyset$ or $\delta(I,w)\subseteq Q_{S}$. In the latter case, CSO of $G$ implies that there is $w'\in P^{-1}P(w)$ such that $\delta(I,w')\cap Q_{NS} \neq \emptyset$.
  \end{proof}

  We now prove the correctness of the construction.
  \begin{thm}
    The DES $G$ with a single observable event is CSO w.r.t. $Q_S$, $Q_{NS}$, $P$ iff $G'$ is K-SO w.r.t. $Q'_S$, $Q'_{NS}$, $P'$.
  \end{thm}
  \begin{proof}
    Assume that $G$ is CSO. We show that if $st \in L(G')$ with $|P'(t)|\le$~K and $s$ leads $G'$ to a secret state from which $t$ can be generated, then there are strings $s'$ and $t'$ with $P'(s')=P'(s)$ and $P'(t')=P'(t)$ such that $s'$ leads $G'$ to a non-secret state from which $t'$ can be generated.

    If $s$ leads $G'$ to a secret state $q_s$, which is also a state of $G$, then CSO of $G$ implies that there is $s'$ with $P(s')=P(s)$ leading $G$, and hence $G'$, to a non-secret state, $q_{ns}$. Therefore, for any extension $t$ of $s$ from state $q_s$, $t'=ua^{|P'(t)|}$ is an extension of $s'$ from $q_{ns}$ with $P'(t')=P'(t)$.

    If $s$ leads $G'$ to state $q_2^\star$, we have two cases.
    If $L(G)$ is infinite, there is $s''\in L(G)$ with $P(s'')=P'(s)$. By Lemma~\ref{lemmaA}, there is $s'\in L(G)$ with $P(s')=P(s'')$ that leads $G$, and hence $G'$, to a non-secret state, say $q_{ns}$. Then, $t'= ua^{|P'(t)|}$ is an extension of $s'$ from $q_{ns}$ with $P'(t')=P'(t)$.
    If $L(G)$ is finite, then $s=s_1us_2$, and $s''=s_1 s_2$ satisfies $P(s'')=P'(s)$. If $s''$ leads $G'$ to state $q_0^\star$, any extension $t$ of $s$ from state $q_2^\star$ is an extension of $s'=s''$ from state $q_0^\star$. If $s''\in L(G)$, then, by Lemma~\ref{lemmaA}, there is $s'$ with $P(s')=P(s'')$ leading $G$ to a non-secret state, from which $t'=ua^{|t|}$ can be generated.
    Altogether, $G'$ is K-SO.

    On the other hand, if $G$ is not CSO, there is $w \in L(G)$ such that $\delta(I,P^{-1}P(w))\cap Q_{S}\neq \emptyset$ and $\delta(I,P^{-1}P(w))\cap Q_{NS}= \emptyset$. In particular, $q_1^\star \notin \delta'(I,P'^{-1}P'(w))\cap Q_{NS}'$. Since $w\in L(G)$, $|P(w)|\leq m$, the maximal number of $a$'s in the strings of $L(G)$, and hence $q_0^\star \notin \delta'(I,P'^{-1}P'(w))\cap Q_{NS}'$. Altogether, $\delta'(I,P'^{-1}P'(w))\cap Q_{NS}' = \emptyset$, and therefore $G'$ is not CSO, neither K-SO.
  \end{proof}

\subsection{Transforming K-SO to CSO}\label{redKSOtoCSO}
  The problem of deciding K-step opacity consists of a DES $G=(Q, \Sigma, \delta, I)$, secret states $Q_S$, non-secret states $Q_{NS}$, and projection $P\colon \Sigma^*\rightarrow \Sigma_o^*$.
  From $G$, we first construct a DES $G''=(Q\cup Q^+\cup Q^-,\Sigma\cup\{@\},\delta',I)$ by creating two disjoint copies of $G$, denoted by $G^+$ and $G^-$, with the state sets $Q^+=\{q^+ \mid q \in Q\}$ and $Q^-=\{q^- \mid q \in Q\}$, and with an additional observable event $@$ that connects $G$ to $G^+$ and $G^-$ by the transitions $(q,@,q^+)$, for every $q\in Q_S$, and $(q,@,q^-)$, for every $q\in Q_{NS}$. The secret states are $Q_S''=Q^+$ and the non-secret states are $Q_{NS}''=Q\cup Q^-$, see Figure~\ref{fig:k-so-cso-1}.

  \begin{figure}
    \centering
    \includegraphics[align=c,scale=.7]{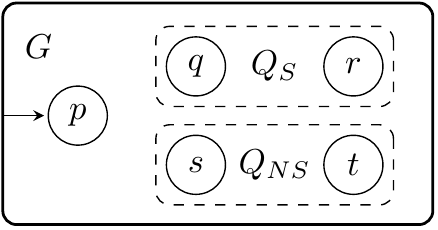}
    $\Longrightarrow$
    \includegraphics[align=c,scale=.7]{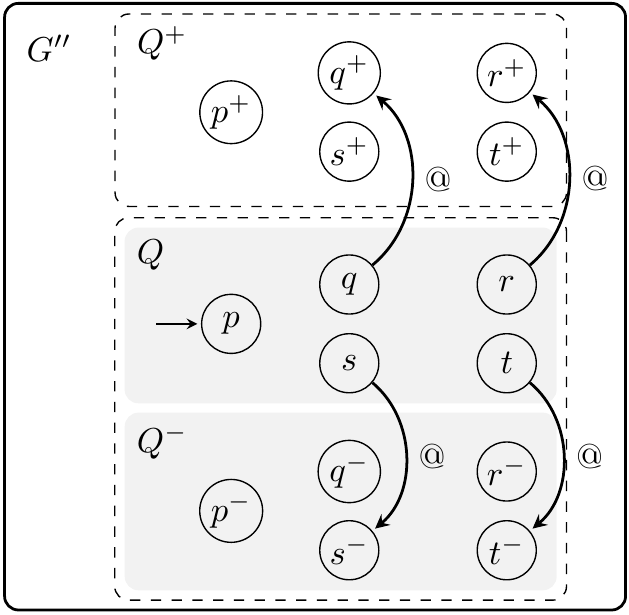}
    \caption{Automaton $G''$ of the first step transforming K-SO to CSO.}
    \label{fig:k-so-cso-1}
  \end{figure}

  The idea of the construction is that if $G$ is K\mbox{-}SO, and hence CSO, then $G$ is in a non-secret state whenever it is in a secret state. Therefore, being in a secret (and hence also in a non-secret) state, the new $@$-transitions move the computation to both new copies $G^+$ and $G^-$. In these copies, we verify that if $G$ can make $k$ steps from the secret state (in $G^+$), it can also make $k$ steps from the corresponding non-secret state (in $G^-$). This is verified using current-state opacity, by considering the states of $G^+$ secret and of $G^-$ non-secret, which requires that every move in $G^+$ must be accompanied by a move in $G^-$.

  Notice that $G''$ can be constructed in polynomial-time using no neutral states. The construction of $G''$ is already suitable to verify INSO of $G$ by checking CSO of $G''$.

  \begin{thm}[Transforming INSO to CSO]\label{prop:inso-cso}
    The DES $G$ is INSO w.r.t. $Q_S$, $Q_{NS}$, and $P$ iff $G''$ is CSO w.r.t. $Q_S''$, $Q_{NS}''$, and $P''\colon (\Sigma\cup\{@\})^* \to (\Sigma_o\cup\{@\})^*$.
  \end{thm}
  \begin{proof}
    Assume that $G$ is INSO. To show that $G''$ is CSO, we consider a string $w$ such that $\delta''(I,w)\cap Q_S'' \neq \emptyset$, and show that there is $w'$ such that $P''(w)=P''(w')$ and $\delta''(I,w')\cap Q_{NS}'' \neq \emptyset$.
    Since $Q_S''=Q^+$, string $w$ is of the form $w_1@w_2$. By construction, there is a secret state $q \in \delta(I,w_1)\cap Q_S$ in $G$ such that $q^+ \in \delta''(I,w_1@) \cap Q_S''$ in $G''$, and $w_2$ is generated from $q^+$. Therefore, we can generate $w_2$ from state $q$ in $G$, that is, $\delta(\delta(I,w_1)\cap Q_S, w_2) \neq \emptyset$, and infinite-step opacity of $G$ implies that there is $w_1'w_2' \in L(G)$ such that $P(w_1)=P(w_1')$, $P(w_2)=P(w_2')$, and $\delta (\delta(I,w_1')\cap Q_{NS}, w_2') \neq \emptyset$.
    If we set $w'=w_1'@w_2'$, then $P''(w)=P''(w')$ and we have that $\emptyset \neq \delta''(\delta''(I,w_1'@)\cap Q_{NS}'', w_2') \subseteq Q_{NS}''$, which completes this part of the proof.

    If $G$ is not INSO, then there is a string $st \in L(G)$ such that $\delta(\delta(I,s)\cap Q_S,t) \neq \emptyset$ and $\delta (\delta(I,s')\cap Q_{NS}, t') = \emptyset$ for every $s't' \in L(G)$ with $P(s)=P(s')$ and $P(t)=P(t')$. Taking $s@t \in L(G'')$, we obtain that $\emptyset \neq \delta''(\delta''(I,s@)\cap Q_S'',t) = \delta''(I,s@t) \subseteq Q_S''$ and, for every $s'@t' \in L(G'')$ with $P''(s@t)=P''(s'@t')$, we have that $\delta''(I,s'@t') \cap Q_{NS}'' = \delta'' (\delta''(I,s'@)\cap Q_{NS}'', t') = \emptyset$, and hence $G''$ is not CSO.
  \end{proof}

  Although $G''$ can verify INSO of $G$ by checking CSO of $G''$, $G''$ is not suitable to verify K-SO in general; indeed, $G''$ verifies any number of steps from the visited secret state rather than at most K steps. To overcome this issue, we extend the construction by adding a counter that allows us to count up to K observable events from a visited secret state. To this aim, we use the automaton $\A_\textrm{K}$ constructed in Appendix~\ref{appB}. Recall that $\A_\textrm{K}$ is of size polynomial in the logarithm of K, that the unique initial state of $\A_\textrm{K}$ is denoted by $q_0$, and that the observer of $\A_\textrm{K}$ has a unique path of length K consisting solely of non-marked states, while all the other states are marked.

  However, the automata $G$, $G^+$, $G^-$ are over the alphabet $\Sigma$, while $\A_\textrm{K}$ is over $\Gamma$, which is disjoint from $\Sigma$. Therefore, we change the alphabets of the automata to $\Sigma' = \Sigma \cup (\Sigma_o \times \Gamma)$ as follows. In $G^+$ and $G^-$, we replace every {\em observable\/} transition $(p,\sigma,q)$ by $|\Gamma|$ transitions $(p,(\sigma,\gamma),q)$, for every $\gamma \in \Gamma$, and denote the results by $\tilde G^+$ and $\tilde G^-$. Similarly, in $\A_\textrm{K}$, we replace every transition $(p,\gamma,q)$ by $|\Sigma_o|$ transitions $(p,(\sigma,\gamma),q)$, for every observable $\sigma \in \Sigma_o$, and denote the result by $\tilde \A_\textrm{K}$.
  If we simplify the strings of the form $(a,a')(b,b')$ as $(ab,a'b')$, then the language $P(L(\tilde G^+)) = \{ (w,w') \in (\Sigma_o \times \Gamma)^* \mid w\in P(L(G^+)),|w|=|w'|\}$. Similarly for $\tilde G^-$ and $\tilde \A_\textrm{K}$.

  For a moment, we admit neutral states, and construct the NFA $G'''$ as a disjoint union of $G$, $\tilde G^+$, $\tilde G^-$, and $\tilde \A_\textrm{K}$, together with the transitions $(q,@,q^+)$ and $(q,@,q_0)$, for every $q\in Q_S$, where $q_0$ is the initial state of $\tilde \A_\textrm{K}$, and by $(q,@,q^-)$, for every $q\in Q_{NS}$. The secret states are $Q_S'''=Q^+$ and the non-secret states are $Q_{NS}'''=Q^- \cup \{\text{marked states of } \tilde \A_\textrm{K}\}$. The other states are {\em neutral}.

  The construction transforms the K-SO problem of $G$ to the CSO problem of $G'''$, as we show below. Since the transformation is polynomial in both the system size and the encoding (logarithm) of K, it improves our recent result~\cite{BalunMasopust2021}.

  \begin{figure}
    \centering
    \includegraphics[scale=.55]{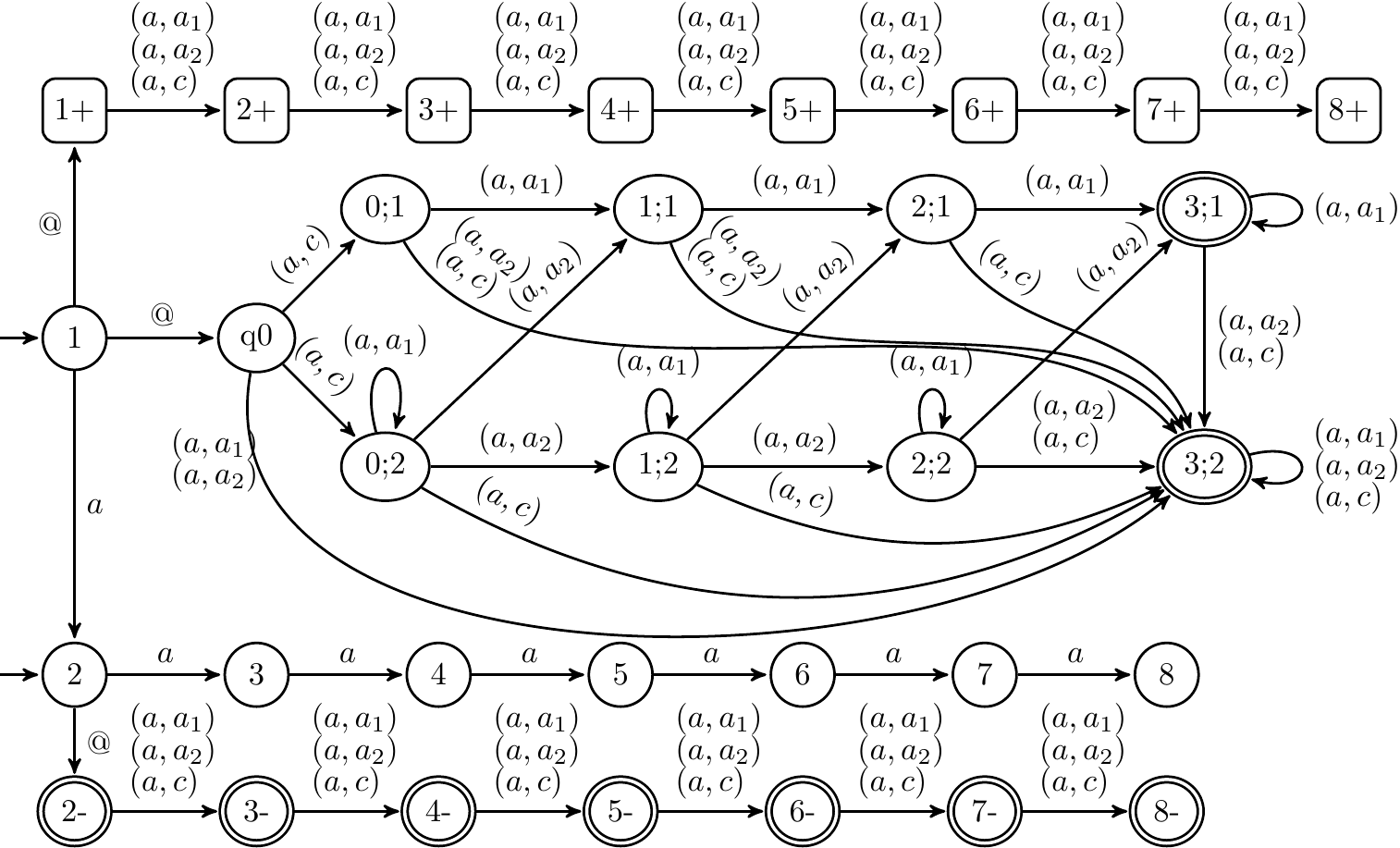}
    \caption{DES $G'''$ of the transformation $6$-SO to CSO with neutral states; secret states are squared and non-secret states are marked.}
    \label{priklad1-CSO}
  \end{figure}

  To illustrate the construction, we transform the 6-SO problem of $G=(\{1,\ldots,8\},\{a\},\delta,\{1,2\})$ with the transitions $\delta(i,a)=\{i+1\}$, $i=1,\ldots,7$, $Q_{S}=\{1\}$, and $Q_{NS}=\{2\}$. Notice that $G$ is 6-SO, since we can make 6 steps from both states 1 and 2. To encode K~$=6$, the transformation uses NFA $\A_6 = \A_{2,2}$ (see Appendix~\ref{appB}), and results in $G'''$ depicted in Figure~\ref{priklad1-CSO}, where all non-secret states are marked. The minimized observer of $G'''$ is shown in Figure~\ref{priklad1-obs}. Since every state of the observer reachable by a string containing $@$ is marked, it has to contain a non-secret state of $G$, that is, $G'''$ is CSO.

  If we remove state $8$ from $G$ together with the corresponding transitions, then $G$ is not 6-SO, since we can make six steps from the secret state 1, but only five steps from the corresponding non-secret state 2. The transformation results in $G'''$ that coincides with the automaton of Figure~\ref{priklad1-CSO} disregarding states $8$, $8^+$, $8^-$, and corresponding transitions. The minimized observer is shown in Figure~\ref{priklad2-obs}, where state $4$, corresponding to state $\{7^+,(2;1),(2;2)\}$ consisting of secret states of $G$, is reachable by $@(a,c)(a,a_1)(a,a_1)(a,a_2)(a,a_1)(a,a_2)$, that is, $G'''$ is not CSO.

  \begin{figure}
    \centering
    \includegraphics[scale=.55]{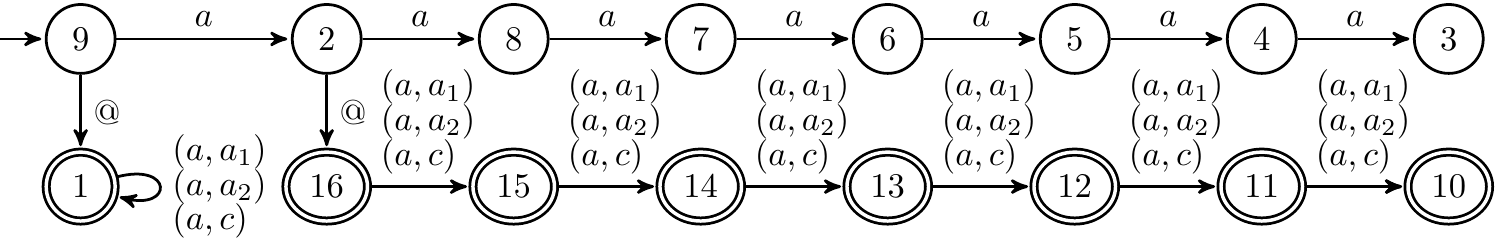}
    \caption{The minimized observer of $G'''$ of Figure~\ref{priklad1-CSO}.}
    \label{priklad1-obs}
  \end{figure}

  \begin{figure}
    \centering
    \includegraphics[scale=.6]{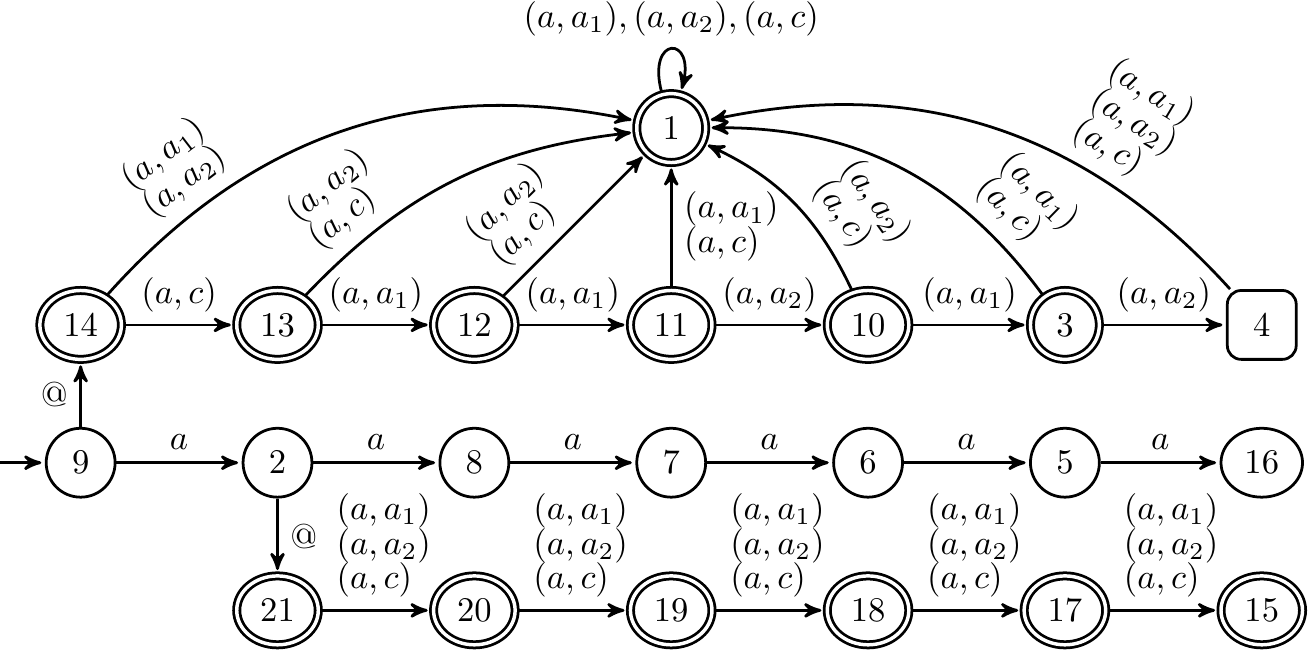}
    \caption{The minimized observer of $G'''$ of Figure~\ref{priklad1-CSO} disregarding states $8$, $8^+$, $8^-$, and corresponding transitions.}
    \label{priklad2-obs}
  \end{figure}

  \begin{thm}[K-SO to CSO with neutral states]
    The DES $G$ is K-SO w.r.t. $Q_S$, $Q_{NS}$, and $P$ iff $G'''$ is CSO w.r.t. $Q_S'''$, $Q_{NS}'''$, and $P'''\colon (\Sigma'\cup\{@\})^* \to (\Sigma_o\cup\{@\}\cup \Sigma_o\times\Gamma)^*$.
  \end{thm}
  \begin{proof}
    Assume that $G$ is K-SO. We show that $G'''$ is CSO. To this end, we consider a string $w$ such that $\delta'''(I,w)\cap Q_S''' \neq \emptyset$, and show that there is a string $w'\in P'''^{-1}P'''(w)$ such that $\delta'''(I,w')\cap Q_{NS}''' \neq \emptyset$.
    Since $Q_S'''=Q^+$, string $w$ is of the form $w_1 @ w_2$ and, by construction, $\delta(I,w_1)$ contains a secret state of $G$ from which $w_2$ can be generated.
    If $|P(w_2)|\le$~K, then K-SO of $G$ implies the existence of $w_1'w_2' \in L(G)$ such that $P(w_1')=P(w_1)$, $P(w_2')=P(w_2)$, and $\delta (\delta(I,w_1')\cap Q_{NS}, w_2') \neq \emptyset$; that is, there is a non-secret state $q\in \delta(I,w_1')$ from which $w_2'$ can be generated, reaching a state $r$.
    Then, for $w'=w_1'@(w_2',x)$, where $x$ is a prefix of the unique string not accepted by $\A_{\textrm{K}}$ of length $|P(w_2')|$, we obtain that $\delta'''(I,w')\cap Q_{NS}''' \neq \emptyset$, since the non-secret state $r^- \in Q^-$ is reachable from state $q^-$ in $G'''$ by $(w_2',x)$.
    If $|P(w_2)|> K$, every string $w_1'@(w_2',y)\in (\Sigma'\cup\{@\})^*$ is such that $y$ is accepted by $\A_{\textrm{K}}$, and hence $\delta'''(I,w_1'@(w_2',y))\cap Q_{NS}'''\neq \emptyset$. Thus, $G'''$ is CSO.

    Assume that $G$ is not K-SO, that is, there exists $st \in L(G)$ such that $|P(t)|\le K$, $\delta(\delta(I,s)\cap Q_S,t) \neq \emptyset$ and, for every $s'\in P^{-1}P(s)$ and $t'\in P^{-1}P(t)$, $\delta (\delta(I,s')\cap Q_{NS}, t') = \emptyset$. Then, in particular, $\delta'''(I,s@) \cap Q_{S}''' \neq \emptyset$.
    If $\delta(I,s') \cap Q_{NS} = \emptyset$, then $\delta'''(I,s'@) \cap Q_{NS}''' = \emptyset$, and hence $G'''$ is not CSO.
    If $\delta(I,s') \cap Q_{NS} = Z \neq \emptyset$, we consider any string $s'@(t',y) \in L(G''')$, where $y$ is a prefix of the unique string not accepted by $\A_\textrm{K}$, which exists because $|y| = |P(t')| \le$~K. Then, $(t',y)$ is not accepted by $\tilde \A_\textrm{K}$, and hence $\delta'''(I,s'@(t',y)) \cap Q_{NS}''' = \delta'''( [\delta'''(I,s'@)\cap Q^-], (t',y)) = \delta'''(Z^-, (t',y)) = \emptyset$, where $Z^- = \{ z^- \mid z\in Z\}$, because $(t',y)$ is not generated in $G'''$ from a state of $Z^-$, since $t'$ cannot be generated in $G$ from any $z\in Z$. Again, $G'''$ is not CSO.
  \end{proof}

  Finally, to transform K-SO to CSO without using neutral states, we make all states of $\tilde G^+$ both initial and marked, and synchronize the computations of $\tilde G^+$ and $\tilde \A_{\textrm{K}}$ by their synchronous product $\tilde G^+ \| \tilde \A_\textrm{K}$.
  Now, we construct a DES $G'$ as a disjoint union of $G$, $\tilde G^-$, and $\tilde G^+ \| \tilde \A_\textrm{K}$, connected together by transitions $(q,@,(q^+,q_0))$, for every $q\in Q_S$, and $(q,@,q^-)$, for every $q\in Q_{NS}$. The secret states of $G'$ are the non-marked states of $\tilde G^+ \| \tilde \A_\textrm{K}$. All the other states are non-secret.

  This transformation can be done in polynomial time in the system size and the binary encoding of K. How to reduce the number of observable events (in all transformations of this section) is discussed in Appendix~\ref{appA}.

  \begin{figure}
    \centering
    \includegraphics[angle=90,origin=c,scale=.49]{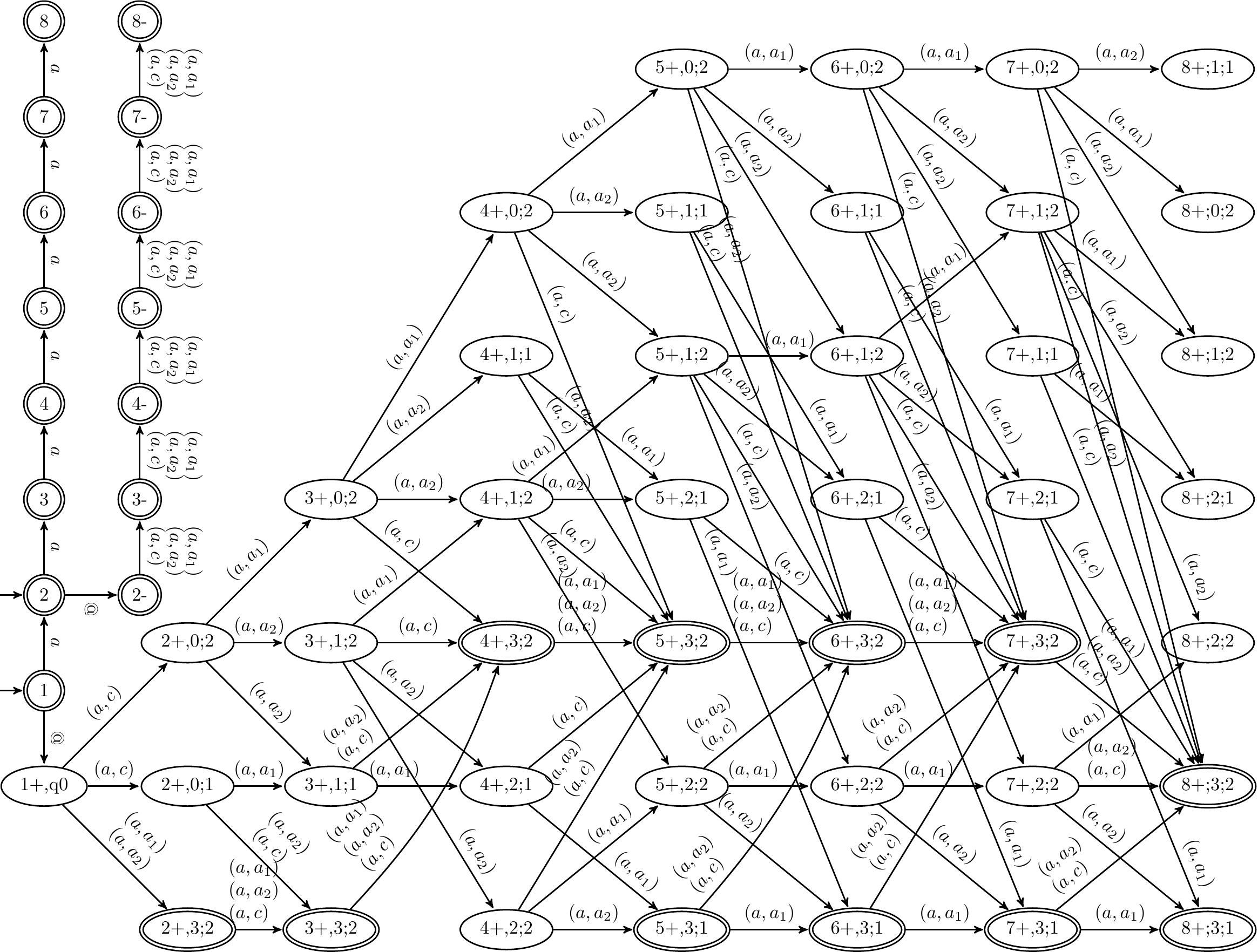}
    \caption{DES $G'$ with a relevant part of $\tilde G^+ \| \tilde \A_6$; non-secret states are marked, other states are secret.}
    \label{priklad6-CSO}
  \end{figure}

  To illustrate the construction, we again transform the 6\mbox{-}SO problem $G=(\{1,\ldots,8\},\{a\},\delta,\{1,2\})$ with state $1$ secret and other states non-secret, and $\delta(i,a)=\{i+1\}$, $i=1,\ldots,7$. The transformation results in $G'$ depicted in Figure~\ref{priklad6-CSO}, using again the NFA $\A_6$. The minimized observer of $G'$ is depicted in Figure~\ref{priklad6-obs}. Since every state of the observer reachable by a string containing $@$ is marked, it has to contain a non-secret state of $G$, that is, $G'$ is CSO.

  If we remove state $8$ from $G$ together with the corresponding transitions, the transformation results in the DES $G'$ that coincides with the NFA of Figure~\ref{priklad6-CSO} without states containing $8$, $8^+$, $8^-$, and the corresponding transitions. The minimized observer is shown in Figure~\ref{priklad5-obs}, where state $4$, abbreviating the state $\{(7^+,(2;1)),(7^+,(2;2))\}$ consisting of secret states of $\tilde G^+ \| \tilde \A_\textrm{K}$, is reachable by the string $@(a,c)(a,a_1)(a,a_1)$ $(a,a_2)(a,a_1)(a,a_2)$, that is, $G'$ is not CSO.

  \begin{figure}
    \centering
    \includegraphics[scale=.5]{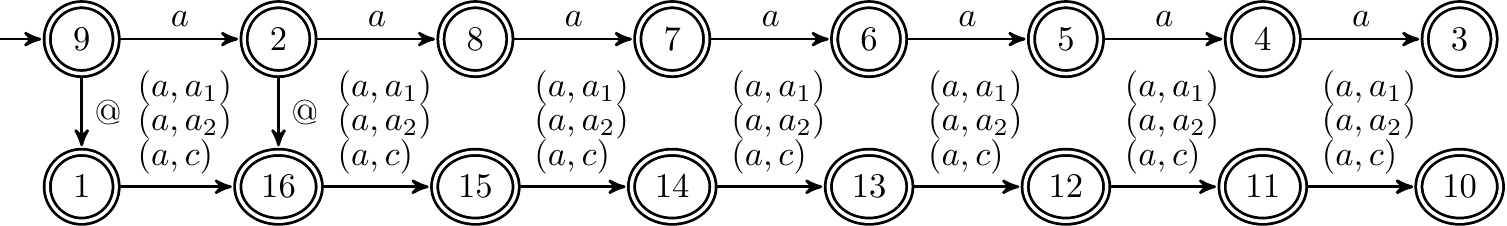}
    \caption{The minimized observer of $G'$.}
    \label{priklad6-obs}
  \end{figure}

  \begin{figure}
    \centering
    \includegraphics[scale=.5]{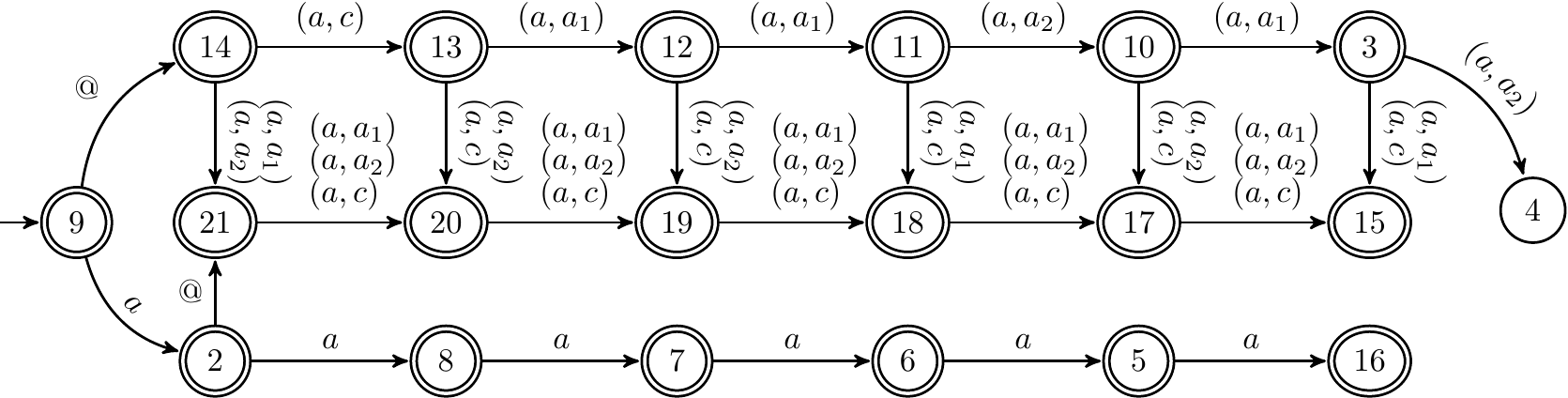}
    \caption{The minimized observer of $G'$ of Figure~\ref{priklad6-CSO} disregarding states containing $8$, $8^+$, $8^-$, and corresponding transitions.}
    \label{priklad5-obs}
  \end{figure}

  \begin{thm}[K-SO to CSO without neutral states]\label{thm:kso-cso}
    The DES $G$ is K-SO w.r.t. $Q_S$, $Q_{NS}$, and $P$ iff $G'$ is CSO w.r.t. $Q_S'$, $Q_{NS}'$, and $P'\colon (\Sigma'\cup\{@\})^* \to (\Sigma_o\cup\{@\}\cup \Sigma_o\times\Gamma)^*$.
  \end{thm}
  \begin{proof}
    Assume that $G$ is K-SO. We show that $G'$ is CSO. To this end, we consider a string $w$ such that $\delta'(I,w)\cap Q_S' \neq \emptyset$, and show that there exists a string $w'\in P'^{-1}P'(w)$ such that $\delta'(I,w')\cap Q_{NS}' \neq \emptyset$.
    Since $Q_S'$ consists of non-marked states of $\tilde G^+ \| \tilde \A_\textrm{K}$, string $w$ is of the form $w_1 @ w_2$ and, by construction, $\delta(I,w_1)$ contains a secret state of $G$, from which $w_2$ can be generated.
    If $|P(w_2)|\le$~K, then K-SO of $G$ implies the existence of $w_1'w_2' \in L(G)$ such that $P(w_1')=P(w_1)$, $P(w_2')=P(w_2)$, and $\delta (\delta(I,w_1')\cap Q_{NS}, w_2') \neq \emptyset$; that is, there is a non-secret state $q\in \delta(I,w_1')$ from which $w_2'$ can be generated, reaching a state $r$.
    Then, for $w'=w_1'@(w_2',x)$, where $x$ is a prefix of the unique string not accepted by $\A_{\textrm{K}}$ of length $|P(w_2)|$, we obtain that $\delta'(I,w')\cap Q_{NS}' \neq \emptyset$, since the non-secret state $r^- \in Q^-$ is reachable from state $q^-$ in $G'$ by $(w_2',x)$.
    If $|P(w_2)|> K$, every string $w_1'@(w_2',y)\in (\Sigma'\cup\{@\})^*$ is such that $y$ is accepted by $\A_{\textrm{K}}$, and hence $(w_2',y)$ is accepted by $\tilde G^+ \| \tilde \A_\textrm{K}$ because $L(\tilde G^+ \| \tilde \A_\textrm{K}) = L(\tilde G^+) \parallel L(\tilde \A_\textrm{K})$ and $(w_2',y)$ belongs to both $L(\tilde G^+)$ and $L(\tilde \A_\textrm{K})$. Therefore, $\delta'(I,w_1'@(w_2',y))\cap Q_{NS}'\neq \emptyset$, and $G'$ is CSO.

    Assume that $G$ is not K-SO, that is, there exists $st \in L(G)$ such that $|P(t)|\le K$, $\delta(\delta(I,s)\cap Q_S,t) \neq \emptyset$ and, for every $s'\in P^{-1}P(s)$ and $t'\in P^{-1}P(t)$, $\delta (\delta(I,s')\cap Q_{NS}, t') = \emptyset$. This implies that $\delta'(I,s@) \cap Q_{S}' \neq \emptyset$.
    If $\delta(I,s') \cap Q_{NS} = \emptyset$, then $\delta'(I,s'@) \cap Q_{NS}' = \emptyset$, and hence $G'$ is not CSO.
    If $\delta(I,s') \cap Q_{NS} = Z \neq \emptyset$, we take any $s'@(t',y) \in L(G')$, where $y$ is a prefix of the unique string not accepted by $\A_\textrm{K}$, which exists because $|y| = |P(t')| \le$~K. Then, $(t',y)$ is not accepted by $\tilde G^+ \| \tilde \A_\textrm{K}$, and hence $\delta'(I,s'@(t',y)) \cap Q_{NS}' = \delta'( [\delta'(I,s'@)\cap Q^-], (t',y)) = \delta'(Z^-, (t',y)) = \emptyset$, where $Z^- = \{ z^- \mid z\in Z\}$, because $(t',y)$ is not generated in $G'$ from a state of $Z^-$, since $t'$ cannot be generated in $G$ from any $z\in Z$. Again, $G'$ is not CSO.
  \end{proof}

  Again, we provide a direct transformation for systems with one observable event, where we do not admit neutral states; see our recent work~\cite{BalunMasopust2021} for systems admitting neutral states.

  The K-SO problem for systems with one observable event consists of a DES $G=(Q,\Sigma,\delta,I)$ with $\Sigma_o=\{a\}$, secret states $Q_{S}$, non-secret states $Q_{NS} = Q \setminus Q_{S}$, and projection $P\colon\Sigma^*\to \{a\}^*$.
  We denote the number of states of $G$ by $n$, and determine (in linear time) whether $P(L(G))$ is finite.

  If $P(L(G))$ is finite, we verify K-SO of $G$ in linear time by checking the subsets of states $\delta(I,P^{-1}(a^k))$, for $k \le n-1$.
  If $G$ is K-SO, and hence CSO, we set $Q_{S}'=Q_S$ and $Q_{NS}'=Q_{NS}$.
  If $G$ is not K-SO, we set $Q_{NS}'=\emptyset$ and $Q_S'=Q$.

  If $P(L(G))$ is infinite, we define $Q_{NS}'= \{ q \in Q_{NS} \mid \varphi(q) = K\}$, where $\varphi\colon Q \rightarrow \{0,\ldots, K\}$ assigns to state $q$ the maximal $k\in \{0,\ldots, K\}$ of observable steps possible from $q$. Formally, $\varphi(q) = \max \{ k\in \{0,\ldots,K\} \mid \delta(q,P^{-1}(a^k))\neq \emptyset \}$. The secret states are $Q_S'=Q\setminus Q_{NS}'$.

  \begin{thm}[K-SO to CSO with a single observable event]
    The DES $G$ with a single observable event is K-SO w.r.t. $Q_S$, $Q_{NS}$, and $P$ iff $G$ is CSO w.r.t. $Q'_S$, $Q'_{NS}$, and $P$.
  \end{thm}
  \begin{proof}
    Assume that $G$ is K-SO. If $P(L(G))$ is finite, then $G$ is CSO. If $P(L(G))$ is infinite, then, for every $w\in L(G)$, there is a state $q\in \delta(I, P^{-1}P(w))$ such that $\varphi(q)=K$. Since $G$ is K-SO, for every secret state $q_{s}\in \delta(I, P^{-1}P(w))$, there is a non-secret state $q_{ns}\in \delta(I, P^{-1}P(w))$ such that $\varphi(q_{ns})\geq\varphi(q_{s})$. Therefore, there is a non-secret state $q_{ns}'\in \delta(I, P^{-1}P(w))$ such that $\varphi(q_{ns}')=K$, which means that $q_{ns}'\in Q_{NS}'$, and hence $G$ is CSO w.r.t. $Q_{NS}'$, $Q_S'$, and $P$.

    Assume that $G$ is not K-SO. If $P(L(G))$ is finite, then $G$ is not CSO. If $P(L(G))$ is infinite, there is $w\in L(G)$ and a secret state $q_s \in \delta(I, P^{-1}P(w))$ such that $\varphi(q_{s}) > \varphi(q_{ns})$ for every non-secret state $q_{ns} \in \delta(I, P^{-1}P(w))$. Therefore, $\varphi(q_{ns}) < K$ for every $q_{ns} \in \delta(I, P^{-1}P(w)) \cap Q_{NS}$, and hence $\delta(I, P^{-1}P(w)) \cap Q_{NS}' = \emptyset$, which shows that $G$ is not CSO w.r.t. $Q_{S}'$, $Q_{NS}'$, and $P$.
  \end{proof}

\section{Conclusions}
  We designed a new algorithm verifying K-step opacity with better complexity than that of existing algorithms. In addition, our complexity does not depend on K. We compared K-step opacity with current-state opacity and infinite-step opacity, and provided new transformations among these notions that do not use neutral states, that are polynomial w.r.t. both the size of the system and the binary encoding of K, preserve determinism (see Appendix~\ref{appC}), and the resulting systems of which do not have more observable events than the input systems.

\appendices

\section{Reducing the Number of Observable Events}\label{appA}
  We now discuss how to reduce the number of observable events in systems with at least three observable events without affecting the property of being CSO. This construction is a modification of the construction of Balun and Masopust~\cite{BalunMasopust2021}.

  For an NFA $G=(Q,\Sigma,\delta,I,F)$, an alphabet $\Gamma_o\subseteq\Sigma_o$ with at least three events, and a binary encoding $e\colon \Gamma_o \to \{0,1\}^k$ of the events of $\Gamma_o$, where $k \le \lceil \log_2(|\Gamma_o|) \rceil + 1$, we define the NFA $G'=(Q',(\Sigma-\Gamma_o)\cup\{0,1\},\delta',I,F)$ as follows. We replace every transition $(p,a,q)$ with $a\in \Gamma_o$ and $e(a)=b_1b_2\cdots b_k \in \{0,1\}^k$ by $k$ transitions
  \[
    (p,b_1,p_{b_1}), (p_{b_1},b_2,p_{b_1b_2}),\ldots, (p_{b_1\cdots b_{k-1}},b_k,q)
  \]
  where $p_{b_1},\ldots,p_{b_1\cdots b_{k-1}}$ are states added to the state set of $G'$ as non-secret states. These states are created when needed for the first time, and reused later during the replacements. Figure~\ref{binEncoding} illustrates the replacement of three observable events $a_1,a_2,a_3$ with the encoding $e(a_1)=00$, $e(a_2)=01$, and $e(a_3)=10$. Notice that $G'$ can be constructed from $G$ in polynomial time.
  
  \begin{figure}
    \centering
    \includegraphics[scale=.7]{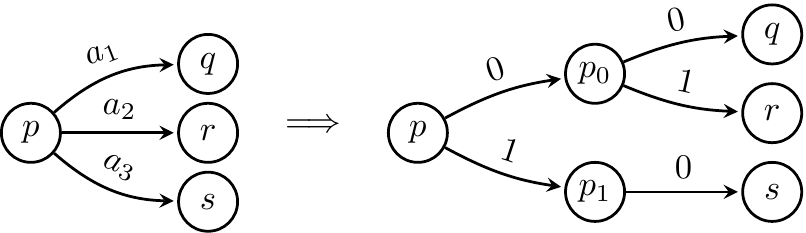}
    \caption{Replacement of observable events $a_1,a_2,a_3$ encoded $e(a_1)=00$, $e(a_2)=01$, and $e(a_3)=10$, and new states $p_0$ and $p_1$.}
    \label{binEncoding}
  \end{figure}

  \begin{lem}\label{binEnc}
    System $G$ is CSO w.r.t. $Q_S$, $Q_{NS}$, and $P\colon \Sigma^*\to\Sigma_o^*$ iff $G'$ is CSO w.r.t. $Q_{S}$, $Q_{NS}\cup (Q'-Q)$, and $P'\colon [(\Sigma-\Gamma_o)\cup\{0,1\}]^* \to [(\Sigma_o-\Gamma_o)\cup \{0,1\}]^*$.
  \end{lem}
  \begin{proof}
    To show that $G$ is CSO iff $G'$ is CSO, we need to show that $P(L_{S})\subseteq P(L_{NS})$ iff $P'(L'_{S})\subseteq P'(L'_{NS})$~\cite{BalunMasopust2021}, where
    \begin{itemize}
      \item $L_S=L_m(Q,\Sigma,\delta,I,Q_S)$,
      \item $L_{NS}=L_m(Q,\Sigma,\delta,I,Q_{NS})$,
      \item $L'_S=L_m(Q',(\Sigma-\Gamma_o)\cup\{0,1\},\delta',I,Q_S)$, and
      \item $L'_{NS}=L_m(Q',(\Sigma-\Gamma_o)\cup\{0,1\},\delta',I,Q_{NS}')$, for $Q_{NS}' = Q_{NS} \cup (Q'-Q)$.
    \end{itemize}

    We define a morphism $f\colon \Sigma^* \to ((\Sigma-\Gamma_o)\cup\{0,1\})^*$ such that $f(a)=e(a)$ for $a\in \Gamma_o$, and $f(a)=a$ for $a\in\Sigma-\Gamma_o$. By the definition of $e$ and the construction of $G'$, $w \in L(G)$ iff $f(w) \in L(G')$. In particular, $P(w)\in P(L_{S})$ iff $P'(f(w)) \in P'(L'_{S})$, and $P(w)\in P(L_{NS})$ iff $P'(f(w)) \in P'(L'_{NS})$. Therefore, if $P'(L'_{S})\subseteq P'(L'_{NS})$ then $P(L_{S})\subseteq P(L_{NS})$. On the other hand, assume that $P(L_{S})\subseteq P(L_{NS})$, and consider any $P'(x) \in P'(L'_{S})$. Then, $P'(x)$ is of the form $P'(f(y))$ for some $y\in L_{S}$, and $P(y)\in P(L_{S})\subseteq P(L_{NS})$ implies that $P'(x)=P'(f(y)) \in P'(L'_{NS})$.
  \end{proof}

\section{Logarithmic Encoding of a K-Step Counter}\label{appB}
  In this appendix, we construct an NFA $\A_\textrm{K}$ of size polynomial in the logarithm of K such that the observer of $\A_\textrm{K}$ has a unique path of length K consisting solely of non-marked states, while all the other states are marked. This path plays the role of a K-step counter that is essential in the transformation from K-SO to CSO of Section~\ref{redKSOtoCSO}.
  To construct the automaton $\A_\textrm{K}$, we make use of NFAs $\A_{k,n}$, for every $k,n\ge 1$, that can be constructed in time polynomial in $k$ and $n$ and that are similar to NFAs we used earlier~\cite{Krtzsch2017}, though we need to adjust them.

  \begin{lem}\label{exprponfas}
    For every integers $k,n\ge 1$, there is an NFA $\A_{k,n}$ with $n$ events and $n(k+2)$ states, such that $\A_{k,n}$ accepts all strings except for all prefixes of a unique string $W_{k,n}$, which is a string of length $\binom{k+n}{k}-1$.
  \end{lem}
  \begin{proof}
    For $k,n\geq 1$, we define $W_{k,n}$ over $\Sigma_n = \{a_1,\ldots, a_n\}$ by setting $W_{k,1} = a_1^k$, $W_{1,n} = a_1a_2\cdots a_n$, and
    \begin{align}
      W_{k,n} & = W_{k,n-1}\, a_{n}\, W_{k-1,n}\nonumber\\
        & = W_{k,n-1}\, a_n\, W_{k-1,n-1}\, a_n\, W_{k-2,n}\label{wkn}\\
        & = W_{k,n-1}\, a_n\, W_{k-1,n-1}\, a_n\, \cdots\, a_n\, W_{1,n-1}\, a_n\,. \nonumber
    \end{align}
    The construction is illustrated in Table~\ref{tableWords}.
    \begin{table}\centering
      \ra{1.2}
      \caption{Strings $W_{k,n}$ used in the proof of Lemma~\ref{exprponfas}.}
      \label{tableWords}
      \begin{tabular}{@{}llll@{}}\toprule
        $k\backslash n$ & 1 & 2 & 3 \\
          \midrule
        1 & $a_1$    & $a_1a_2$                      & $a_1a_2a_3$\\
        2 & $a_1^2$  & $a_1^2 a_2 a_1 a_2$           & $a_1^2 a_2 a_1 a_2 a_3 a_1a_2a_3$\\
        3 & $a_1^3$  & $a_1^3 a_2 a_1^2 a_2 a_1 a_2$ & $a_1^3 a_2 a_1^2 a_2 a_1 a_2 a_3 a_1^2 a_2 a_1 a_2 a_3 a_1a_2a_3$\\
        \bottomrule
      \end{tabular}%
    \end{table}
    The length of $W_{k,n}$ is $\binom{k+n}{k}-1$, and $a_n$ appears exactly $k$ times in $W_{k,n}$~\cite{dlt15}. For defining $\A_{k,n}$, it is useful to set $W_{k,n}=\eps$ whenever $kn=0$.

    We construct an NFA $\A_{k,n}$ over $\Sigma_n$ marking $\Sigma_n^*-\overline{\{W_{k,n}\}}$. For $k\ge 0$, $\A_{k,1}$ is the minimal DFA marking $\{a_1\}^* - \overline{\{a_1^k\}}$, consisting of $k+2$ states of the form $(i;1)$, see Figure~\ref{rpoNFAn1}, together with the given transitions.
    State $(k+1;1)$ is marked, state $(0;1)$ is initial.

    \begin{figure}[b]
      \centering
      \includegraphics[scale=.55]{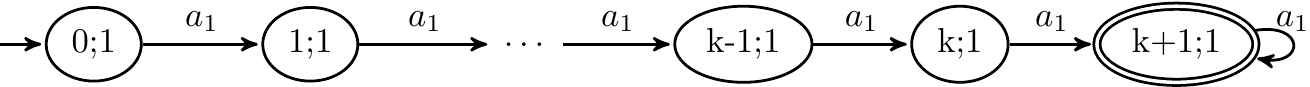}
      \caption{The NFA $\A_{k,1}$ with $k+2$ states.}
      \label{rpoNFAn1}
    \end{figure}

    \begin{figure}
      \centering
      \includegraphics[scale=.55]{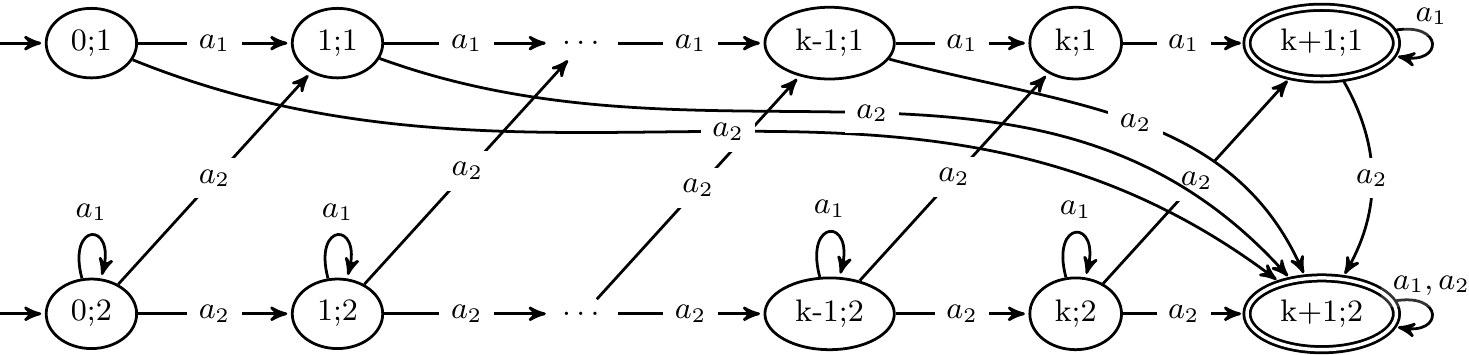}
      \caption{The NFA $\A_{k,2}$ with $2(k+2)$ states.}
      \label{rpoNFAn2}
    \end{figure}
    Given $\A_{k,n-1}$, we construct $\A_{k,n}$ from $\A_{k,n-1}$ by adding $k+2$ states $(0;n),(1;n),\ldots,(k+1;n)$, where $(0;n)$ is added to initial, and $(k+1;n)$ to final states, see Figure~\ref{rpoNFAn2} illustrating the construction for $n=2$; $\A_{k,n}$ has $n(k+2)$ states. We call the state $(k+1,n)$ {\em maximal}.
    Additional transitions of $\A_{k,n}$ consist of four groups:
    \begin{enumerate}
      \item\label{r1} Self-loops $(i;n)\xrightarrow{a_j}(i;n)$ for every $i\in\{0,\ldots,k+1\}$ and $a_j\in\{a_1,\ldots,a_{n-1}\}$;
      \item\label{r2} Transitions $(i;n)\xrightarrow{a_n}(i+1;n)$ for every $i\in\{0,\ldots,k\}$, and the self-loop $(k+1;n)\xrightarrow{a_n} (k+1;n)$;
      \item\label{r3} Transitions $(i;n)\xrightarrow{a_n}(i+1;m)$ for every $i\in\{0,\ldots,k\}$ and $m\in\{1,\ldots,n-1\}$;
      \item\label{r4} Transitions $(i;m)\xrightarrow{a_n}(k+1;n)$ for every state $(i;m)$ of $\A_{k,n-1}$ with $i\neq k$.
    \end{enumerate}
    The additional states of $\A_{k,n}$ and transitions (\ref{r1}) and (\ref{r2}) ensure marking of all strings containing more than $k$ events $a_n$.
    The transitions (\ref{r3}) and (\ref{r4}) ensure marking of all strings in $(\Sigma_{n-1}^* a_n)^{i+1} L(\A_{k-(i+1),n-1})a_n \Sigma_n^*$ for which the string between the $(i+1)$-st and the $(i+2)$-nd occurrence of $a_n$ is not of the form $W_{k-(i+1),n-1}$, and hence not a correct substring of $W_{k,n} = W_{k,n-1} a_n \cdots a_n W_{k-(i+1),n-1} a_n \cdots\allowbreak a_n W_{1,n-1} a_n$.
    The transitions (\ref{r4}) ensure that all strings with a prefix $w a_n$, where $w$ is any string from $\Sigma_{n-1}^*-\{W_{k,n-1}\}$, are marked.
    Together, these conditions ensure that $\A_{k,n}$ marks every string that is not a prefix of $W_{k,n}$.

    It remains to show that $\A_{k,n}$ does not mark any prefix of $W_{k,n}$, which we show by induction on $(k,n)$.
    For $(0,n)$, $n\geq 1$, string $W_{0,n}=\eps$ is not marked by $\A_{0,n}$, since the initial states $(0,m)=(k,m)$ of $\A_{0,n}$ are not marked.
    Likewise, for $(k,1)$, $k\ge 0$, we find that the prefixes of $W_{k,1}=a_1^k$ are not marked by $\A_{k,1}$ (cf.~Figure~\ref{rpoNFAn1}).
    For the inductive case $(k,n)\ge (1,2)$, where $\le$ is the standard product order, we assume that $\A_{k',n'}$ marks no prefix of $W_{k',n'}$ for any $(k',n') < (k,n)$ and that $W_{k',n'}$ leads $\A_{k',n'}$ only to states of the form $(k;\cdot)$.
    Then, $W_{k,n} = W_{k,n-1} a_n W_{k-1,n}$, and no prefix of $W_{k,n-1}$ is marked by $\A_{k,n-1}$ by induction.
    In addition, there is no transition under $a_n$ from a state $(k;m)$ with $m\neq n$ in $\A_{k,n}$.
    Therefore, if a prefix of $W_{k,n}$ is marked by $\A_{k,n}$, it must be marked in a run starting from the initial state $(0;n)$.
    Since $W_{k,n-1}$ contains no $a_n$, we find that $\A_{k,n}$ can only reach states $\delta((0;n), W_{k,n-1} a_n) =\{(1;m)\mid 1\le m \le n\}$ after generating $W_{k,n-1} a_n$, which are the initial states of $\A_{k-1,n}$.
    By induction, $\A_{k-1,n}$ marks no prefix of $W_{k-1,n}$, and hence no prefix of $W_{k,n}$ is marked by $\A_{k,n}$.
  \end{proof}

  \begin{figure}
    \centering
    \includegraphics[scale=.55]{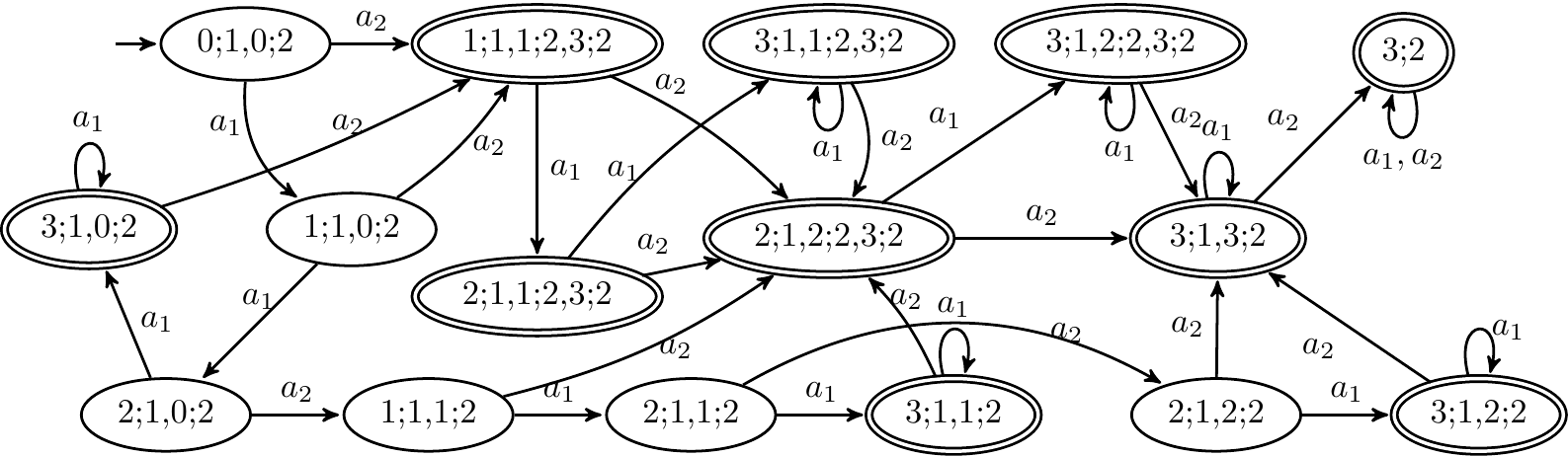}
    \caption{The observer of the NFA $\A_{2,2}$.}\label{figA22obs}
  \end{figure}

  To illustrate the construction, we consider $k=n=2$. Then, $W_{2,2}=a_1^2a_2a_1a_2$, the NFA $\A_{2,2}$ has 8 states, and the observer of $\A_{2,2}$, shown in Figure~\ref{figA22obs}, contains a unique path of length $\binom{4}{2}-1=5$ consisting solely of non-marked states while all the other states are marked.

  We now show how to use the NFAs $\A_{i,j}$ to construct an automaton $\A_\textrm{K}$ of size polynomial in the logarithm of K such that the observer of $\A_\textrm{K}$ has a unique path consisting solely of non-marked states, while all the other states are marked. For simplicity, and without loss of generality, we only use the automata of the form $\A_{j,j}$.

  Since $\binom{2n+2}{n+1}= \frac{4n+2}{n+1} \binom{2n}{n}$ and $\binom{2n}{n} \le 4^n$, every natural number $\textrm{K}$ can be expressed as
  \[
    \textrm{K} = b_n \binom{2n}{n} + b_{n-1} \binom{2n-2}{n-1} + \cdots + b_1 \binom{2}{1} + b_0
  \]
  for some $n\le \lceil\log_4(\textrm{K}+1)\rceil$ and $b_i \in \{0,1,2,3\}$, $i=0,\ldots,n$. This expression is not unique, e.g., $\textrm{K}=2$ can be expressed as $b_0=2$, or $b_1=1$ and $b_0=0$.

  For every $b_i$, $i=n,\ldots,0$, we create $b_i$ copies of $\A_{i,i}$ over $\Sigma_i=\{a_1\ldots,a_i\}$, which results in a sequence of automata $\B_1,\ldots,\B_\ell$. We take a new event $c\notin\Sigma_n$ and connect all the automata $\B_1,\ldots,\B_\ell$ to a single automaton $\A_\textrm{K}$ by $c$\mbox{-}tran\-si\-tions as follows.
  For $j=1,\ldots,\ell-1$, we add a $c$\mbox{-}tran\-si\-tion from every non-marked state of $\B_j$ to every initial state of $\B_{j+1}$; from all the other states, the $c$-transition goes to the maximal state of $\B_\ell$.
  Finally, we add a new state, $q_0$, which is the only initial state of the automaton $\A_\textrm{K}$, $c$-transitions from $q_0$ to all initial states of $\B_1$, and transitions under all the other events to the maximal state of $\B_\ell$; see an illustrative example below.

  Then, the observer of $\A_\textrm{K}$ has a unique path consisting of non-marked states along the string
  \[
    (c W_{n,n})^{b_n} (c W_{n-1,n-1})^{b_{n-1}} \cdots (c W_{0,0})^{b_0}
  \]
  of length $\textrm{K} = b_n \binom{2n}{n} + b_{n-1} \binom{2n-2}{n-1} + \cdots + b_0 \binom{0}{0}$, and the other states are marked. Since every $\B_j$ is of size polynomial in $n$, we obtain that $\A_\textrm{K}$ is of size polynomial in the logarithm of $\textrm{K}$ and its observer has a unique path of length $\textrm{K}$ consisting solely of non-marked states, with all the other states marked.

  \begin{lem}
    For every natural number $\textrm{K}$, there is an automata $\A_\textrm{K}$ of size polynomial in $O(\log \textrm{K})$ such that the observer of $\A_\textrm{K}$ has a unique path of length $\textrm{K}$ consisting solely of non-marked states, and with all the other states marked.
    \qed
  \end{lem}

  \begin{figure}
    \centering
    \includegraphics[scale=.49]{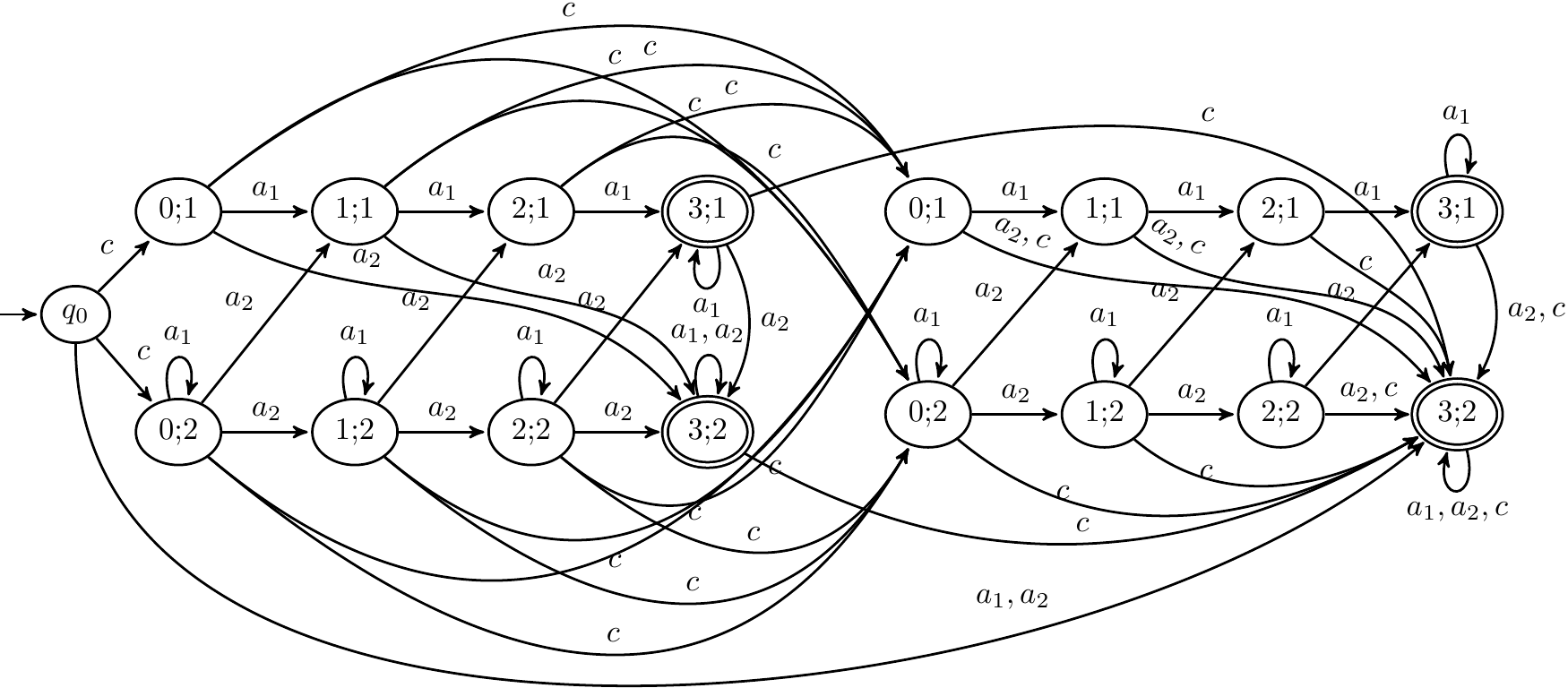}
    \caption{Example for $\textrm{K}=12$, which gives $a_2=2$, the automaton $\A_{12}$ consisting of two copies of $\A_{2,2}$.}
    \label{figEX1}
  \end{figure}

  \begin{figure}
    \centering
    \includegraphics[scale=.8]{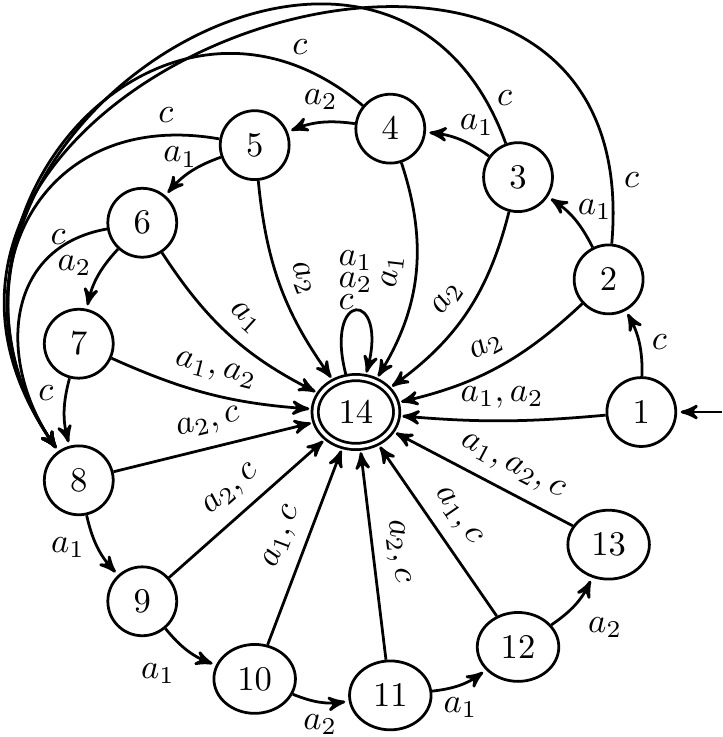}
    \caption{The min. DFA of the observer with the unique path of length $12$.}
    \label{figEX2}
  \end{figure}

  For an illustration, consider $\textrm{K}=12 = 2 \binom{4}{2} + 0 \binom{2}{1} + 0$. We create two copies of $\A_{2,2}$ and connect them by $c$-transitions as shown in Figure~\ref{figEX1}. The observer with the unique path of non-marked states of length $\textrm{K}=12$ is shown in Figure~\ref{figEX2}.

\section{Preserving Determinism}\label{appC}
  In this section, we show how to make an NFA deterministic without affecting the property of being K-step opaque, for any K~$\in \mathbb{N}_\infty$, by adding a few unobservable events.

  From an NFA $G=(Q,\Sigma,\delta,I,F)$, we construct a DFA $G'=(Q',\Sigma',\delta',I,F)$ as follows. For every state $p$ and an event $a$ with $|\delta(p,a)|>1$, we replace every transition $(p,a,q)$ of $G$ with two transitions
  $(p,u,p')$ and $(p',a,q)$, where $p'$ is a new state and $u$ is a new unobservable event (neither $p'$ nor $u$ are reused), see Figure~\ref{fig_det} for an illustration. The secret status of the new state $p'$ is set according to the status of state $p$, that is, $p'$ is secret iff $p$ is. Notice that $G'$ can be constructed from $G$ in polynomial time.

  \begin{figure}
    \centering
    \includegraphics[scale=.7]{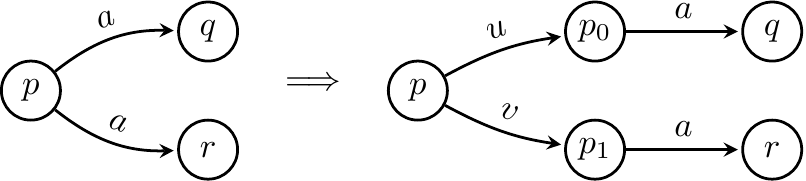}
    \caption{Determinization of an NFA.}
    \label{fig_det}
  \end{figure}

  \begin{lem}
    System $G$ is K-SO w.r.t. $Q_S$, $Q_{NS}$, and $P\colon \Sigma^*\to\Sigma_o^*$ iff $G'$ is K-SO w.r.t. $Q_{S}'$, $Q_{NS}'$, and $P'\colon \Sigma'^* \to \Sigma_o^*$.
  \end{lem}
  \begin{proof}
    Indeed, the number of observable steps from a state of $G$ is preserved in $G'$. Thus, we need to show that $G'$ is CSO iff $G$ is. However, every newly added state $p'$ is reachable by an unobservable event from its original state $p$, and hence $p'$ is contained in every state of the observer that contains $p$; and vice versa, because every path to state $p'$ goes through state $p$ in $G'$. Therefore, if a state of the observer contains a secret state $p'$ and a non-secret state $r'$, then it also contains the original secret state $p$ and the original non-secret state $r$. That is, $G$ is K-SO iff $G'$ is K-SO.
  \end{proof}

\newpage
\bibliographystyle{IEEEtran}
\bibliography{mybib}

\end{document}